\documentclass[a4paper,8pt]{article}
\usepackage{a4}
\usepackage{graphicx}
\usepackage{texab}
\usepackage{amsmath,amssymb}
\usepackage{amsthm}
\usepackage{bm}
\usepackage{natbib}

\newtheoremstyle{thm}
{15pt}
{5pt}
{\itshape}
{}
{\bf}
{}
{ }
{}

\theoremstyle{thm} 

\newtheorem{Alg}{Algorithm}
\newtheorem{lemma}{Lemma}
\newtheorem{theorem}{Theorem}
\newtheorem{corollary}{Corollary}

\newtheorem{claim}{Claim}

\newtheorem{Ass}{Assumption}

\begin{document}
\bibliographystyle{apalike} 

\title{Estimation for High-Dimensional Linear Mixed-Effects Models Using
  $\ell_1$-Penalization}
\author{J\"urg Schelldorfer, Peter B\"uhlmann and Sara van de Geer \footnote{The research is
supported by the Swiss National Science Foundation (grant
no. 20PA21-120043/1). The authors thank the members of the DFG-SNF Forschergruppe 916 for many stimulating discussions.}\\Seminar f\"ur Statistik
  \\ETH Z\"urich\\8092 Z\"urich
}



\maketitle

\begin{abstract}
We propose an $\ell_1$-penalized estimation procedure
for high-dimensional linear mixed-effects models. The models are useful
whenever there is a grouping structure among high-dimensional observations,
i.e. for clustered data. We prove a consistency and an oracle optimality
result and we develop an algorithm with provable numerical
convergence. Furthermore, we demonstrate the performance of the method on
simulated and a real high-dimensional data set.
\end{abstract}

\vspace*{.3in}


\section{Introduction} \label{Sec1}

\subsection{High-dimensional statistical inference: some known results for convex loss functions} 

Substantial progress has been achieved over the last decade in
high-dimensional statistical inference where the number of parameters $p$ is
allowed to be of much larger order than sample size $n$. To fix ideas,
suppose we focus on 
estimation of a $p$-dimensional parameter $\bm\beta_0$ based on $n$ noisy
observations where $p \gg n$. Although
such a problem is ill-posed in general, it can be accurately
solved if the underlying true structure of $\bm\beta_0$ is sparse. Here, sparsity
may be measured in terms of the $\ell_r$-norm $\|\bm\beta\|_r = (\sum_{j=1}^p
|\beta_j|^r)^{1/r}\ (0 \le r < \infty)$. Very roughly speaking,
high-dimensional statistical inference is possible, in 
the sense of leading to reasonable accuracy or asymptotic consistency, if 
\begin{eqnarray*}
\log(p) \cdot \mbox{sparsity}(\bm\beta_0)^{\alpha} \ll n,
\end{eqnarray*}
where typically $\alpha =2$ (cf. formula (\ref{f1})) or $\alpha = 1$ (cf. formula
(\ref{f2})), and assuming that the underlying (e.g. regression) design
behaves reasonably.   

A lot of attention has been devoted to high-dimensional linear models
\begin{equation*}
\bm{y} = \bm{X} \bm\beta_0 + \bm\varepsilon ,
\end{equation*}
with $n \times p$ design matrix $\bm{X}$ and $p \gg n$. A very popular and
powerful estimation method is the Lasso, proposed by
\cite{Tibs96}. It is an acronym for
  Least Absolute Shrinkage and Selection Operator and the name is
  indicating that the method is
  doing some variable selection in the sense that some of the regression
  coefficient estimates are exactly zero. Among the main reasons
why it has become very popular for high-dimensional estimation problems are 
its statistical accuracy for prediction and variable selection coupled with
its computational feasibility which involves convex optimization only. The
latter is in sharp contrast to exhaustive variable selection based on least
squares estimation whose computational complexity is in general exponential
in $p$. The statistical properties of the Lasso in high-dimensional
settings have been worked out in numerous articles. Without (essentially) a 
condition on the design $\bm{X}$, the Lasso satisfies:
\begin{equation}\label{f1}
\|\bm{X}(\hat{\bm\beta} - \bm\beta_0)\|_2^2/n = O_P(\|\bm\beta_0\|_1
\sqrt{\log(p)/n})
\end{equation}
where $O_P(\cdot)$ is with respect to $p \ge n \to \infty$
\citep{BuhlGeer10}. That is, if the
model is sparse with $\|\bm\beta_0\|_1 \ll \sqrt{\log(p)/n}$, we obtain
consistency. Such kind of a result has been proved by
\cite{Greenshtein04}. Later, optimality has been established where
(\ref{f1}) is improved to 
\begin{eqnarray}\label{f2}
&\|\bm{X}(\hat{\bm\beta} - \bm\beta_0)\|_2^2/n = O_P(s_0
\xi^{-2} \log(p)/n),\nonumber\\ 
\mbox{and furthermore}\nonumber\\ \ &\|\hat{\bm\beta}- \bm\beta_0\|_r  =
O_P(s_0^{1/r} \xi^{-2} \sqrt{\log(p)/n}),\ r \in \{1,2\},
\end{eqnarray}
where $s_0$ equals the number of non-zero coefficients and
$\xi^2$ denotes a restricted eigenvalue of the design matrix $\bm{X}$
\citep{BuhlGeer10}. The rate in (\ref{f2}) is optimal up to the $\log(p)$
factor and the restricted 
eigenvalue $\xi^2$: oracle least squares estimation where the relevant
variables would be known would have rate $O_P(s_0/n)$. 
We emphasize that for obtaining optimal convergence rates as in (\ref{f2}), we
need to make some assumptions on the design that $\xi^2$ is not getting too
small as $p \ge n \to \infty$, something we do not require in
(\ref{f1}). Works dealing with various aspects around (\ref{f2}) include
\cite{Bunea07}, \cite{Geer08}, \cite{Zhang08}, \cite{Meins09} and
\cite{Bickel09}.\\
A quite different problem is variable selection for inferring the true
underlying active set $S_0 = \{\ 1 \le k \le p :\ \beta_{0,k} \neq
0\}$. A simple estimator is $\hat{S} = \{\ 1 \le k \le p:\ \hat{\beta}_k \neq 0\}$ where no
significance testing is involved. \cite{Mein06} show for the Lasso that
under the so-called 
neighborhood stability condition for the design, the Lasso does consistent
variable selection in the sense that 
\begin{equation}\label{f3}
\mathbb{P}[\hat{S} = S_0] \to 1\ (p \ge n \to \infty),
\end{equation} 
assuming that the non-zero coefficients in $S_0$ are sufficiently large in
absolute value, e.g. $\min_{k \in S_0} |\beta_{0,k}| \gg s_0
\xi^{-2} \sqrt{\log(p)/n}$ which is the rate in (\ref{f2}) for $r=1$.  
The neighborhood stability condition is equivalent to the irrepresentable
condition used in \cite{Zhao06}, and they are both sufficient and (essentially)
necessary for consistent model selection as in (\ref{f3}). Unfortunately,
the neighborhood stability and the irrepresentable condition are rather
restrictive and many designs $\bm{X}$ would violate them. In case of
(weaker) restrictive eigenvalue conditions, 
one still has the variable screening property for the Lasso 
\begin{equation}\label{f4}
\mathbb{P}[\hat{S} \supseteq S_0] \to 1\ (p \ge n \to \infty),
\end{equation} 
again assuming that the non-zero coefficients in $S_0$ are sufficiently
large in absolute value. Formula (\ref{f4}) says that the Lasso does
not miss a relevant variable from $S_0$; in addition, for the Lasso, the
cardinality $|\hat{S}| \le \min(n,p)$ and hence, for $p \gg n$, we achieve
a huge dimensionality reduction in (\ref{f4}). The adaptive Lasso, proposed
by \cite{Zou06} is a two-stage method which achieves (\ref{f3}) under
weaker restrictive eigenvalue assumption than the irrepresentable condition
\citep{HuangMaZhang08,GeerBuhlZhou10}. We summarize the basic facts in
Table \ref{table1}.
\begin{table}[!h]
\footnotesize
\begin{center}
\caption{\textit{Properties of the Lasso and required conditions to
  achieve them}} \label{table1}
\vspace{0.2cm}
\begin{tabular}{l|l|l}
\hline
property & design condition & size of non-zero coeff.
\\
\hline \hline
consistency as in (\ref{f1}) & no requirement & no requirement 
\\ 

\hline
fast convergence rate as in (\ref{f2}) & restricted eigenvalue & no
requirement \\

\hline
variable selection as in (\ref{f3}) & neighborhood stability & sufficiently
large\\ 
 & $\Leftrightarrow$ irrepresentable cond. & \\

\hline
variable screening as in (\ref{f4}) & restricted eigenvalue & sufficiently
large \\ \hline
\end{tabular}
\end{center}
Restricted eigenvalue assumption is weaker than the neighborhood
stability or irrepresentable condition \citep{GeerBuhl09}. For the
adaptive Lasso: variable selection as in (\ref{f3}) can be achieved under
restricted eigenvalue conditions.
\end{table}
\normalsize
Moreover, everything essentially holds in an analogous
way when using the Lasso in generalized linear models, i.e. $\ell_1$-norm
penalization of the 
negative log-likelihood \citep{Geer08}. Finally, we note that
\cite{Bickel09} prove equivalent theoretical behaviour of the Lasso and
the Dantzig selector \citep{Candes07} in terms of (\ref{f2}), exemplifying
that properties like (\ref{f2}) hold for other estimators 
than the Lasso as well. \\

Having some variable screening property as in (\ref{f4}), we can reduce the
false positive selections by various methods, besides the adaptive Lasso
mentioned above, including also stability selection \citep{MeinBuhl10} based
on subsampling or via
assigning p-values \citep{WassRoed09,MeinMeiBuhl09} based on sample
splitting. 

Regarding computation, the Lasso involves convex optimization. Popular
algorithms are based on the homotopy method
\citep{OsbPresTur00} such as LARS \citep{Efro03}. More
recently, it has been argued that the coordinate gradient descent approach
is typically more efficient \citep{Meie08,Wu08,Fried08}. 

\subsection{High-dimensional linear mixed-effects models with 
non-convex loss function} 

The underlying assumption that all observations are independent is not
always appropriate. We consider here linear mixed-effects models
\citep{Lair82,PinJB2000,VerbMole00,Demi04} where
high-dimensional data incorporates a grouping structure with
independent observations  
between and dependence within groups. Mixed-effects models, including
random besides fixed effects, are a popular
extension of linear models in that direction. For 
example, many applications concern 
longitudinal data where the random effects vary between groups and thereby
induce a dependence structure within groups. It is a crucial and
important question how to cope with high-dimensional linear mixed-effects
models. Surprisingly, for this problem, there is no established procedure
which is well understood in terms of statistical properties. 

The main difficulty arises from non-convexity of the negative
log-likelihood function which makes computation and theory very
challenging. We are presenting some methodology, computation and theory for
$\ell_1$-norm penalized maximum likelihood estimation in linear mixed-effects models where the number of fixed effects may be much larger than
the overall sample size but the number of covariance parameters of the
random effects part being small. 
Based on a framework for $\ell_1$-penalization of smooth but non-convex
negative log-likelihood functions \citep{Stad09}, we develop in Section
\ref{Sec3} analogues of (\ref{f1}), (\ref{f2}) and (\ref{f4}), see also Table
\ref{table1}, and some properties of an adaptively $\ell_1$-penalized
estimator. In our view, these are the key properties in high-dimensional
statistical inference in any kind of model. For example, with (\ref{f4}) at
hand, p-values for single fixed-effects coefficients could be constructed
along the lines of \cite{MeinMeiBuhl09}, controlling the familywise error
or false discovery rate (but we do not apply such a method in this
paper). Furthermore, we 
design in Section \ref{Sec5} an efficient coordinate gradient descent 
algorithm for linear mixed-effects models which is proved to converge
numerically to a stationary point of the corresponding non-convex
optimization problem. 

We remark that we focus here on the case where it is pre-specified which
covariates are modelled with a random effect and which are not. In some
situations, this is fairly realistic: e.g., a random intercept model is
quite popular and often leads to a reasonable model fit. Without
pre-specification of the covariates 
having a random effect, one could do variable selection based on penalized
likelihood approaches on the level of
random effects: this has been developed from a methodological and
computational perspective by \cite{Bondell10} and \cite{Ibrahim10} for
low-dimensional settings. Addressing 
such problems in the truly high-dimensional scenario is beyond the
scope of this paper. However, we present in Section \ref{Sec7} a real
high-dimensional data problem where some exploratory analysis is used for
deciding which covariates are to be modelled with a random effect. This
example also illustrates empirically that there is a striking improvement
if we incorporate random effects into the model, in
comparison to a high-dimensional linear model fit.\\

The rest of this paper is organised as follows. In Section
\ref{Sec2} we define the $\ell_1$-penalized linear mixed-effects estimator. In
Section \ref{Sec3}, we present the
theoretical results for this estimator before describing
the details of a computational algorithm in Section \ref{Sec5}. After some simulations in
Section \ref{Sec6} we apply the procedure to a real data set. The technical proofs
are deferred to an Appendix in the Supporting Information.

\section{Linear mixed-effects models and $\ell_1$-penalized estimation} \label{Sec2}

\subsection{High-dimensional model set-up}
We assume that the observations are inhomogeneous in the sense that they
are not independent, but grouped. Let $i=1,\ldots ,N$ be the grouping index and
$j=1,\ldots,n_i$ the observation index within a group. Denote by
$N_T=\sum_{i=1}^N n_i$ the total number of observations.
For each group, we observe a $n_i \times 1$ vector of responses
$\bm{y}_i$, and let $\bm{X}_i$ be a $n_i \times p$ fixed-effects design matrix, $\bm{\beta}$ a $p
\times 1$ vector of fixed regression coefficients, $\bm{Z}_i$ a $n_i \times
q$ random-effects design matrix and $\bm{b}_i$ a group-specific vector
of random regression coefficients. \\
Using the notation from \cite{PinJB2000}, the model can be written as
\begin{equation}\label{model2}
\bm{y}_i = \bm{X}_i \bm{\beta} + \bm{Z}_i \bm{b}_i + \bm{\varepsilon}_i \hspace{1cm} i=1,\ldots,N,
\end{equation}
assuming that
\begin{enumerate}
\item [$i)$]  $\bm{\varepsilon}_i \sim
\mathcal{N}_{n_i}(\bm{0}, \sigma^2\bm{I}_{n_i})$ and uncorrelated for
  $i=1,\ldots,N$, 
\item [$ii)$] $\bm{b}_i \sim \mathcal{N}_q(\bm{0},\bm\Psi)$ and uncorrelated for
  $i=1,\ldots,N$, 
\item [$iii)$]
  $\bm{\varepsilon}_1,\ldots,\bm{\varepsilon}_N,\bm{b}_1,\ldots,\bm{b}_N$
  are independent.
\end{enumerate}
Here, $\bm\Psi=\bm\Psi_{\bm\theta}$ is a general covariance matrix where
$\bm\theta$ is an unconstrained set of parameters (with dimension $q^*$)
such that $\bm\Psi_{\bm\theta}$ is positive definite (i.e. by using the Cholesky decomposition). Possible structures for $\bm\Psi$ may be a multiple of
the identity, a diagonal or a general positive definite matrix. We would
like to remark that assumption $i)$ can be generalized to $ i')$ $\bm{\varepsilon}_i \sim
\mathcal{N}_{n_i}(\bm{0}, \sigma^2\bm\Lambda_i)$ with
$\bm\Lambda_i=\bm\Lambda_i(\bm\lambda)$ for a parameter vector
$\bm\lambda$. This generalization still fits into the
theoretical framework presented in Section \ref{Sec3}. Nonetheless, for the
sake of notational simplicity, we restrict ourselves to assumption $i)$.\\ 
As indicated by the index $i$, the $\bm{b}_i$ are
different among the groups. All observations have the coefficient $\bm{\beta}$
in common whereas the value of $\bm{b}_i$ depends on the group that the observation
belongs to. In other words, for each group there are group-specific
deviations $\bm{b}_i$ from the overall effects $\bm{\beta}$. We assume
throughout the paper that the design matrices $\bm{X}_i$ and $\bm{Z}_i$ are deterministic, i.e. fixed design.\\

We allow that the number $p$ of fixed-effects regression
coefficients may be much larger than the total number of observations,
i.e. $N_T \ll p$. Furthermore, the number $q$ of random-effects
variables might be as large as $q \le p$, but the dimension $q^*$ of the
variance-covariance parameters is assumed to be small ($q^* \ll N_T$). We aim at estimating the fixed regression parameter vector $\bm{\beta}$,
the random effects $\bm{b}_i$ and the variance-covariance parameters
$\bm\theta$ and $\sigma^2$. Therefore, $\tilde{\bm\phi}:=(\bm\beta^
T,\bm\theta^T,\sigma^2)^T$ defines the complete parameter vector with at most length $p
+ \frac{q(q+1)}{2} + 1$. From model (\ref{model2}) we deduce that $\bm{y}_1,\ldots,\bm{y}_N$ are
independent and $\bm{y}_i \sim \mathcal{N}_{n_i}({\bm{X}_i
  \bm{\beta},\bm{V}_i(\bm\theta,\sigma^2)})$ with $\bm{V}_i(\bm\theta,\sigma^2)
= \bm{Z}_i\bm\Psi_{\bm\theta}\bm{Z}_i^T + 
\sigma^2 \bm{I}_{n_i}$. Denote the stacked vectors $\bm{y}=(\bm{y}^T_1,\ldots,\bm{y}^T_N)^T$, 
$\bm{b}=(\bm{b}^T_1,\ldots,\bm{b}^T_N)^T$, $\bm\varepsilon=(\bm\varepsilon_1^T,\ldots,\bm\varepsilon_N^T)^T$ and the stacked
matrices $\bm{X}=(\bm{X}_1^T,\ldots,\bm{X}_N^T)^T$, $\bm{Z}=\diag(\bm{Z}_1,\ldots,\bm{Z}_N)$ and
$\bm{V}=\diag(\bm{V}_1,\ldots,\bm{V}_N)$. Then model (\ref{model2}) can be
written as
\begin{equation}
\bm{y} = \bm{X} \bm{\beta} + \bm{Z} \bm{b} + \bm{\varepsilon}
\end{equation}
and the negative log-likelihood is given by
\small
\begin{equation} \label{seven}
-\ell(\tilde{\bm\phi})=-\ell (\bm\beta,\bm\theta,\sigma^2) =
 \frac{1}{2} \Big\{ N_T \log(2 \pi) +
\log|\bm{V}| +  (\bm{y}-\bm{X}\bm{\beta})^T
\bm{V}^{-1}(\bm{y}-\bm{X}\bm{\beta}) \Big\}, 
\end{equation}
\normalsize
where $|\bm{V}|=\det(\bm{V})$.

\subsection{$\ell_1$-penalized maximum likelihood estimator}
Due to the possibly large number of covariates ($N_T \ll p$ setting), we cannot use the classical
maximum likelihood or restricted maximum likelihood approach. Assume that the fixed regression coefficients are sparse in
the sense that many parameters are zero. We then attenuate these difficulties by adding an
$\ell_1$-penalty on the fixed regression coefficients. By doing so, we achieve a sparse
solution with respect to the fixed effects. This leads us to consider the following objective function:
\begin{equation} \label{five}
Q_{\lambda}(\bm\beta,\bm\theta,\sigma^2) :=  \frac{1}{2} \log|\bm{V}| +  \frac{1}{2}(\bm{y}-\bm{X}\bm{\beta})^T
\bm{V}^{-1}(\bm{y}-\bm{X}\bm{\beta}) + \lambda \sum_{k=1}^p
|\beta_k| ,
\end{equation}
where $\lambda$ is a nonnegative
regularization parameter. Consequently, we estimate the fixed regression
coefficient vector $\bm{\beta}$ and the variance components $\bm\theta$ and
$\sigma^2$ by
\begin{equation} \label{opt2}
\hat{\tilde{\bm\phi}} = (\hat{\bm{\beta}},\hat{\bm\theta},\hat{\sigma}^2)=
\argmin_{\bm\beta, \bm\theta,\sigma^2>0,\bm\Psi>0}
Q_{\lambda}(\bm{\beta},\bm\theta,\sigma^2).
\end{equation}
For fixed variance parameters $\bm\theta$ and $\sigma^2$, the minimization with
respect to $\bm{\beta}$ is a convex optimization problem. Since we want to
make use of this convexity (see Section 4), we do not profile the
likelihood function, as usually done in the mixed-effects
model framework \citep{PinJB2000}. However, with
respect to the full parameter vector $\bm\phi$, we have a non-convex objective
function and hence, we have to deal with a non-convex problem. This requires a more general framework in theory as well as in
computation. In the following sections, we discuss how to address this
issue.\\

\subsection{Prediction of the random-effects coefficients}
We predict the random-effects coefficients $\bm{b}_i$, $i=1,\ldots,N$ by
the maximum a posteriori (MAP) principle. Denoting by $f$ the density of the corresponding Gaussian random
variable, we define
\begin{align*}
\tilde{\bm{b}}_i & = \argmax_{\bm{b}_i} f(\bm{b}_i | \bm{y}_1,\ldots,\bm{y}_N,\bm{\beta},\bm\theta,\sigma^2) = \argmax_{\bm{b}_i} f(\bm{b}_i |\bm{y}_i,\bm{\beta},\bm\theta,\sigma^2) \\
& = \argmax_{\bm{b}_i} {\frac{f(\bm{y}_i|\bm{b}_i,\bm\beta,\sigma^2) \cdot
    f(\bm{b}_i|\bm\theta)}{f(\bm{y}_i|\bm\beta,\bm\theta,\sigma^2)}} \\
& = \argmin_{\bm{b}_i} \Bigg\{\frac{1}{\sigma^2} \|\bm{y}_i-\bm{X}_i\bm{\beta}-\bm{Z}_i
\bm{b}_i \|^2 +  \bm{b}_i^T\bm\Psi_{\bm\theta}^{-1}\bm{b}_i    \Bigg\} .
\end{align*}
From this we get $\tilde{\bm{b}}_i = [\bm{Z}_i^T\bm{Z}_i
+\sigma^2\bm\Psi_{\bm\theta}^{-1}]^{-1} \bm{Z}_i^T\bm{r}_i$  where $\bm{r}_i =
(\bm{y}_i-\bm{X}_i\bm{\beta})$ is the (marginal) residual vector. Since the
true values of $\bm{\beta}$, $\bm\theta$ and $\sigma^2$ are unknown, the $\bm{b}_i$'s are predicted by
$\hat{\bm{b}}_i = [\bm{Z}_i^T\bm{Z}_i +\hat{\sigma}^2\bm\Psi_{\hat{\bm\theta}}^{-1}]^{-1} \bm{Z}_i^T\hat{\bm{r}}_i$ with
$\hat{\bm{r}}_i = (\bm{y}_i-\bm{X}_i\hat{\bm{\beta}})$, using the estimates from (\ref{opt2}).

\subsection{Selection of the regularization parameter}
The estimation requires to choose a regularization parameter $\lambda$. We
propose to use the Bayesian Information Criterion (BIC) defined by
\begin{equation} \label{bic}
BIC_{\lambda}:=-2\ell(\hat{\bm{\beta}},\hat{\bm\theta},\hat{\sigma}^2) + \log N_T \cdot
\hat{df}_{\lambda},
\end{equation}
where $\hat{df}_{\lambda}:=|\{1\le k \le p : \hat{\beta}_k \ne
0\}| + \dim(\bm\theta)$ is the sum of the number of the
non-zero fixed regression coefficients and the number of variance-covariance
parameters. The use  of $|\{1\le k \le p :
\hat{\beta}_k \ne 0\}|$ as a measure of the degrees of freedom is motivated by the work of \cite{Zou07} who show that the
expected number of degrees of freedom for the Lasso in a linear model is
given by the number of non-zero estimated coefficients.\\
Obviously, there are other tuning parameter selection methods, for
example cross-validation and AIC-type criteria, among others. Advocating
the BIC as selection criterion is based on our empirical experience that it
performs best in both simulations and real data examples (see Section \ref{Sec6} and
\ref{Sec7}). 

\subsection{Adaptive $\ell_1$-penalized maximum likelihood
  estimator}
Due to the bias of the Lasso, \cite{Zou06} proposed the adaptive
Lasso. For some given weights $w_1,\ldots,w_p$, the 
adaptive $\ell_1$-penalized maximum likelihood estimator has the
following objective function instead of (\ref{five}):
\begin{equation*}
Q_{\lambda}^{w_1,\ldots,w_p}(\bm{\beta},\bm\theta,\sigma^2) := \frac{1}{2} \log|\bm{V}| + \frac{1}{2}(\bm{y}-\bm{X}\bm{\beta})^T
\bm{V}^{-1}(\bm{y}-\bm{X}\bm{\beta}) +
\lambda \sum_{k=1}^p w_k |\beta_k| ,
\end{equation*}
and hence
\begin{equation} \label{opt3}
\hat{\tilde{\bm\phi}} = (\hat{\bm{\beta}},\hat{\bm\theta},\hat{\sigma}^2)=
\argmin_{\bm{\beta}, \bm\theta, \sigma^2>0,\bm\Psi>0}
Q_{\lambda}^{w_1,\ldots,w_p}(\bm{\beta},\bm\theta,\sigma^2) .
\end{equation}
The weights $w_1,\ldots,w_p$ may be calculated from an initial estimator
$\hat{\bm\beta}_{init}$ in (\ref{opt2}) with $w_k:=1/|\hat{\beta}_{init,k}(\lambda)|$ for
$k=1,\ldots,p$. Unless specified otherwise, we employ these weights.

\section{Theoretical Results} \label{Sec3}
In the high-dimensional setting with $p \gg N_T$, the theory for penalized
estimation based on convex loss functions with an $\ell_1$-penalty is well
studied, see for example \cite{Geer08}. From (\ref{five}) and (\ref{opt2})
we see that we are dealing with a non-convex loss function, due to the
variance parameters $\bm\theta$ and $\sigma^2$, and a convex
$\ell_1$-penalty. To the best of our knowledge, only \cite{Stad09} consider high-dimensional non-convex $\ell_1$-penalized smooth likelihood
problems. In this section, we build upon the theory presented in
\cite{Stad09} and extend their results to prove an oracle
inequality for the adaptive $\ell_1$-penalized estimator (\ref{opt3}).\\
We use the following framework and notation. Let $i=1,\ldots,N$ as before
and $n_i \equiv n>1$ the same for all $i$. Denote by $\bm{y}_i \in
\mathcal{Y} \subset \mathbb{R}^n$ the response variable. Let $\bm{X}_i$ be
the fixed covariates in some space $\mathcal{X}^n \subset \mathbb{R}^{n
  \times p}$ and $\bm{Z}_i \subset \bm{X}_i$. The latter can be assumed
without loss of generality, since we can assign to every variable a fixed
effect being equal to zero. Define the parameter $\bm{\phi}^T :=
(\bm{\beta}^T,\bm\theta^T,2 \log \sigma)
=(\bm{\beta}^T,\bm\theta^T,\varrho) = (\bm{\beta}^T,\bm\eta^T) \in \mathbb{R}^{p+q^*+1}$
and denote by $\bm{\phi}_0$ the true parameter vector. For a constant $0 < K < \infty$, consider the parameter space
\begin{equation} \label{211209}
\bm{\Phi} = \{\bm\phi^T=(\bm{\beta}^T,\bm\eta^T) : \sup_{\bm{x} \in \mathcal{X}}|\bm{x}^T\bm{\beta}|
\le K, \|\bm\eta\|_{\infty} \le K, \bm\Psi>0 \} \in \mathbb{R}^{p+q^*+1},
\end{equation}
where $\|\bm\eta\|_{\infty}=\max_{l}|\eta_l|$. We modify the
estimators in (\ref{opt2}) and (\ref{opt3}) by restricting the solution to
lie in the parameter space $\bm\Phi$:
\small
\begin{equation} \label{210909}
\hat{\bm{\phi}} := \argmin_{\bm{\phi} \in \bm\Phi}
\Bigg\{\frac{1}{2} \log|\bm{V}|
+\frac{1}{2} (\bm{y}-\bm{X}\bm{\beta})^T\bm{V}^{-1}(\bm{y}-\bm{X}\bm{\beta}) + \lambda \sum_{k=1}^p |\beta_k|
\Bigg\},
\end{equation}
\begin{equation} \label{261010}
\hat{\bm{\phi}}_{weight} := \argmin_{\bm{\phi} \in \bm\Phi}
\Bigg\{\frac{1}{2} \log|\bm{V}|
+\frac{1}{2} (\bm{y}-\bm{X}\bm{\beta})^T\bm{V}^{-1}(\bm{y}-\bm{X}\bm{\beta}) + \lambda \sum_{k=1}^p w_k|\beta_k|
\Bigg\}.
\end{equation}
\normalsize
Now, let $f_{\bm{\phi},\bm{X}_i,\bm{Z}_i}$ be the Gaussian density for $\bm{y}_i$ with respect to the above parametrization. Since we use the
negative log-likelihood as loss function, the excess risk coincides with
the Kullback-Leibler distance:
\begin{equation}
\mathcal{E}_{\bm{X},\bm{Z}}(\bm{\phi} | \bm{\phi}_0) = \int \log
\Bigg(\frac{f_{\bm{\phi}_0,\bm{X},\bm{Z}}}{f_{\bm{\phi},\bm{X},\bm{Z}}}\Bigg)f_{\bm{\phi}_0,\bm{X},\bm{Z}}
d\mu,
\end{equation}
where $\mu$ denotes the Lebesgue measure, and we define the average excess risk as
\begin{equation*}
 \overline{\mathcal{E}}_{\bm{X}_1,\ldots,\bm{X}_N,\bm{Z}_1,\ldots,\bm{Z}_N}(\bm{\phi} |
 \bm{\phi}_0) =\frac{1}{N} \sum_{i=1}^N
 \mathcal{E}_{\bm{X}_i,\bm{Z}_i}(\bm{\phi}|\bm{\phi}_0).
\end{equation*}
In the sequel, we drop the indices $_{\bm{X},\bm{Z}}$ and
$_{\bm{X}_1,\ldots,\bm{X}_N,\bm{Z}_1,\ldots,\bm{Z}_N}$, respectively.\\

\subsection{Consistency for the $\ell_1$-penalized estimator}
We require only one condition for consistency. It is a condition on the
random-effects design matrices $\bm{Z}_i$.
\begin{Ass} \textbf{} \label{Ass1}
The eigenvalues of
$\bm{Z}^T_i\bm{Z}_i$, denoted by
$\big(\nu_j^{(i)}\big)_{j=1}^q$  for $i=1,\ldots,N$, are bounded: $\nu_{j}^{(i)} 
  \le K < \infty$ for all $i$ and $j$, with $K$ from (\ref{211209}).
\end{Ass}

Now we consider a triangular scheme of observations from (\ref{model2}):
\begin{equation} \label{triarray}
\bm{y}_i = \bm{X}_i \bm{\beta}_N + \bm{Z}_i \bm{b}_i + \bm{\varepsilon}_i
\hspace{1cm} i=1,\ldots,N , 
\end{equation}
where the parameters $\bm{\beta}_N$ and $\bm{\eta}_N$ are allowed to depend on
$N$. We study consistency as $N \to \infty$ but the group size $n$ is
fixed. Moreover, let us use the notation $a\vee b := \max\{a,b\}$.
\begin{theorem} {\bf(Consistency)}  \label{AppA} \\
Consider model (\ref{triarray}) and the estimator (\ref{210909}). Under
Assumption \ref{Ass1} and assuming
\begin{equation*}
\|\bm{\beta}_{0,N}\|_1 = o\Bigg(\sqrt{\frac{N}{\log^4(N)\log(p\vee N)}}\Bigg), \quad \lambda_N = C
\sqrt{\frac{\log^4(N)\log(p\vee N)}{N}}
\end{equation*}
for some $C > 0$, any global minimizer $\hat{\bm{\phi}}$ as in (\ref{210909})
satisfies $\bar{\mathcal{E}}(\hat{\bm{\phi}}|\bm{\phi}_0) =
o_P(1)$ as $N \rightarrow \infty$.
\end{theorem}
A proof is given in the Appendix in the Supporting Information. The condition on $\|\bm{\beta}_{0,N}\|_1$ is
a sparsity condition on the true underlying fixed-effects coefficients.

\subsection{Oracle inequality for the adaptive $\ell_1$-penalized estimator}
We now present an oracle optimality result in non-asymptotic
form for the adaptive estimator (and thereby covering also the non-adaptive
case). Preliminary, we introduce some notation and two further 
assumptions.
\begin{Ass} \textbf{} \label{Ass2}
\begin{itemize}
\item [(a)] Let $\big(\omega_j^{(i)}\big)_{j=1}^n$ be the eigenvalues of
$\bm{Z}_i\bm\Psi \bm{Z}_i^T$ for $i=1,\ldots,N$. At least two eigenvalues
are different, i.e. for all $i$ $ 
\exists  j_1\ne j_2 \in \{1,\ldots,n \}$ such that $\omega_{j_1}^{(i)} \neq \omega_{j_2}^{(i)}$.
\item [(b)] For $i=1,\ldots,N$, the matrices $\bm\Omega_i$ defined by
\begin{equation*}
(\bm\Omega_i)_{r,s}=
  \trace\bigg(\bm{V}_i^{-1}\frac{\partial \bm{V}_i}{\partial
    \phi_{p+r}}\bm{V}_i^{-1}\frac{\partial \bm{V}_i}{\partial
    \phi_{p+s}}\bigg) \quad r,s=1,\ldots,q^*+1
\end{equation*}
 are strictly positive definite.
\end{itemize}
\end{Ass}
\emph{Remark.} In the special case $\bm\Psi=\theta^2\bm{I}$, Assumption
\ref{Ass2} (b) automatically holds.\\
Let $S(\bm{\beta}) = \{1 \le k \le p : \beta_k \ne 0 \}$ be the active
set of $\bm{\beta}$, i.e. the set of non-zero coefficients, and $\bm{\beta}_{\mathcal{K}} =
\{\beta_k : k \in \mathcal{K}  \}$ for $\mathcal{K} \subset \{1,\ldots,p\}$. We denote by
$S_0=S(\bm{\beta}_0)$ the true active set and by $s_0=|S_0|$ its
cardinality.  Write $\bm{X}_i^T = (\bm{x}_{1}^i,\ldots,\bm{x}_{n}^i)$ and
define 
\begin{equation} \nonumber
\bm\Sigma_{N,n}:=\frac{1}{N} \sum_{i=1}^N \sum_{j=1}^n \bm{x}_j^i(\bm{x}_j^i)^T \quad
\in \mathbb{R}^{p \times p}
\end{equation}

\begin{Ass}{\bf(Restricted Eigenvalue Condition)} \label{Ass3} \\ 
There exists a constant $\kappa \ge 1$, such that for all $\bm{\beta} \in
\mathbb{R}^{p}$ satisfying $\|\bm{\beta}_{S_0^c}\|_1 \le 6 \|\bm{\beta}_{S_0} \|_1$
it holds that $\|\bm{\beta}_{S_0}\|_2^2 \le \kappa^2 \bm{\beta}^T \bm\Sigma_{N,n} \bm\beta$.
\end{Ass}
A discussion of this assumption can be found in \cite{Bickel09} and
\cite{GeerBuhl09}. Define
\begin{equation} \label{11.08}
\lambda_0 = M_N \log N  \sqrt{\frac{\log(p \vee N)}{N}} ,
\end{equation}
where $M_N$ is of order $\log N$ and an exact definition is
given in the proof of Theorem \ref{AppA}. For any $T \ge a_1$, let
$\cal J$ be a set defined by the underlying empirical process (see (A.6) in
the Supporting Information). It is
shown in the proof of Theorem \ref{AppA} that the set $\cal J$ has
large probability,
\begin{equation*}
\mathbb{P}[\mathcal{J}] \ge 1-a_2\exp\Big[-\frac{T^2 \log^2 N \log(p \vee N)}{a_3^2} \Big] - \frac{\rho}{\log N}\frac{1}{N^{1-2\varepsilon}}
\end{equation*}
for $N$ sufficiently large and some constants $a_1, a_2, a_3, \varepsilon,
\rho>0$, see Lemma 2 and 3 in the Appendix A (in the Supporting
Information).\\

At this point, we could conclude an oracle result in the way of
\cite{Stad09}. However, we extend that result and present an
oracle inequality involving $\| \hat{\bm\beta} - \bm\beta_0\|_1$ instead of
$\|\hat{\bm\beta}_{S_0^c}\|_1$  for the $\ell_1$-penalized as well as the
adaptive $\ell_1$-penalized estimator.
\begin{theorem} \textbf{(Oracle result)} \label{weight.theorem}\\
Consider the weighted $\ell_1$-penalized estimator (\ref{261010}). Suppose that for some $\delta > 0$,
\begin{equation*}
w_k \begin{cases}
\le 1/ \delta & k \in S_0 , \\
\ge 1/ \delta & k \notin S_0  . \\
\end{cases}
\end{equation*}\\
Under Assumptions \ref{Ass1}, \ref{Ass2} and \ref{Ass3}, and 
for $\lambda \ge 2 T \delta \lambda_0$,
we have on the set $\cal J$ defined in (A.6), 
$$ \bar{\mathcal{E}} ( \hat{\bm\phi}_{weight} \vert \bm\phi_0 ) +
 2( \lambda/\delta  - T  \lambda_0 ) \| \hat{\bm\beta}_{weight} - \bm\beta_0 \|_1 \le
 9 ( \lambda/\delta  + T  \lambda_0 )^2 c_0^2 \kappa^2 s_0 . $$
\end{theorem}
The proof is given in the Appendix B. Application of Theorem
\ref{weight.theorem} with all weights equal to one ($\delta=1$) gives an oracle result
for the $\ell_1$-penalized estimator, which we will use as initial values for the adaptive Lasso procedure.

\begin{corollary}\label{init.corollary}
Let
\begin{equation*}
\hat{\bm\phi}_{init}:= (\hat{\bm\beta}_{init},\hat{\bm\theta}_{init},\hat\varrho_{init}):=
\argmin_{\bm\phi \in \bm\Phi}
Q_{\lambda_{init}}^{1,\ldots,1 }({\bm\beta},\bm\theta,\varrho),
\end{equation*}
be the initial estimator in (\ref{261010}) (i.e., the estimator with all the weights equal to one).
Under Assumptions \ref{Ass1}, \ref{Ass2} and \ref{Ass3}, and 
for $\lambda_{init} \ge 2 T  \lambda_0$,
we have on ${\cal J}$,
\begin{equation}\label{init.equation} {\bar{\cal E}} ( \hat{\bm\phi}_{init} \vert \bm\phi_0 ) +
 2( \lambda_{init}  - T  \lambda_0 ) \| \hat{\bm\beta}_{init} - \bm\beta_0
 \|_1 \le  9 ( \lambda_{init}  + T  \lambda_0 )^2 c_0^2 \kappa^2 s_0 . 
\end{equation}
\end{corollary}
  It is clear that the $\ell_1$-estimation error bound implies a
  bound for the $\ell_{\infty}$ estimation error as well. When the underlying
  true coefficients $\beta_{0,k}$, $ k \in S_0$ are sufficiently much larger in absolute value
  than the $\ell_{\infty}$-estimation error bound, 
   one can perfectly
  distinguish between active and non-active set. This argument is applied in the next
  corollary to the adaptive Lasso with estimated weights.

\begin{corollary}\label{adap.corollary}
Let
$$\hat{\bm\phi}_{adap} :=(\hat{\bm\beta}_{adap},\hat{\bm\theta}_{adap},\hat\varrho_{adap}):=
\argmin_{\bm\phi \in \bm\Phi}
Q_{\lambda_{adap}}^{w_1,\ldots,w_p }({\bm\beta},\bm\theta,\varrho),
$$
be the adaptive estimator with weights $w_k = 1/ |\hat \beta_{{init},k}|$,
$k=1 , \ldots , p $ as in (\ref{210909}). 
Assume that for all $ k \in S_0$, 
\begin{equation}\label{betamin}
|\beta_{0,k} | \ge 2 \delta_{init}, 
\end{equation}
where
$$ \delta_{init} := 
\frac {9 ( \lambda_{init}  + T  \lambda_0 )^2 c_0^2 \kappa^2 s_0 }{
2( \lambda_{init}  - T  \lambda_0 )} \ge \|\hat{\bm\beta}_{init} -
\bm\beta_0 \|_1 $$ is a bound of the $\ell_1$-estimation error of the
initial $\hat{\bm\beta}_{init}$.
Suppose moreover that Assumptions \ref{Ass1}, \ref{Ass2} and \ref{Ass3} are met.
Then,  for $\lambda_{adap} \ge 2 T \delta_{init} \lambda_0$,
and on the set ${\cal J}$,
\begin{flalign} \label{adap.equation} \nonumber
&{\bar{\cal E}} ( \hat{\bm\phi}_{adap} \vert \bm\phi_0 ) +
 2( \lambda_{adap}/\delta_{init}  - T  \lambda_0 ) \| \hat{\bm\beta}_{adap}
 - \bm\beta_0 \|_1 \\
& \le
 9 ( \lambda_{adap}/\delta_{init}  + T  \lambda_0 )^2 c_0^2 \kappa^2 s_0 .
\end{flalign}
\end{corollary}
We call condition (\ref{betamin}) a "betamin" condition.
 It is clearly very restrictive, but allows for an easy derivation
 of the oracle result.
 The "betamin" condition can indeed be substantially refined.
 In 
 \cite{GeerBuhlZhou10}, one can find similar  oracle 
results, and in addition variable selection results, for the
 adaptive Lasso in the linear model, without "betamin" conditions. 
 These results 
 require introducing various versions of restricted eigenvalues
 and sparse eigenvalues, and
 can be generalized to the current setting. 
 Since a full presentation is rather involved, we have
 confined ourselves to the simplest case.
 
 Recall that in (\ref{11.08}), we choose $\lambda_0$ of order $\log^2  N \sqrt {
 \log ( p \vee N) / N} $. When we also choose $\lambda_{init}$ of this order,
 we find, modulo the restricted eigenvalue $\kappa$ and the constants $T$ and $c_0$,
 that the right-hand side of the oracle result (\ref{init.equation}) for the
 initial estimator is of order
 $$ \log^4 N \frac {\log ( p \vee N)} {N} s_0 , $$ 
 and that
 $$\delta_{\rm init} \asymp \log^2  N \sqrt {
\frac { \log ( p \vee N) }{N} } s_0 . $$
 The tuning parameter for the adaptive Lasso can then be taken of order
 $$\lambda_{adap} \asymp \log^4 N \frac {\log ( p \vee N)} {N} s_0 . $$
 The right-hand side (\ref{adap.equation}) of the oracle result for the adaptive
 estimator is then of the same order
  as the one for the initial estimator. 
 
Assuming "betamin" conditions, the results in Corollary
\ref{init.corollary} and \ref{adap.corollary} imply the variable screening
property motivated already in (\ref{f4}). 

\begin{corollary} \textbf{} \label{screen.corollary}
\begin{enumerate}
\item [1)]For the $\ell_1$-penalized (initial) estimator (\ref{210909}) , assume
$$\min_k|\beta_{0,k}| > \delta_{init}=\frac {9 ( \lambda_{init}  + T  \lambda_0 )^2 c_0^2 \kappa^2 s_0 }{
2( \lambda_{init}  - T  \lambda_0 )}.$$ Then, under the assumptions of
Corollary \ref{init.corollary}, on the set $\mathcal{J}$, $$S_0 \subset
\hat{S}_{init} = \{1 \le k \le p : \hat{\beta}_{init,k} \ne 0 \}.$$ 
\item [2)] For the adaptive $\ell_1$-penalized estimator in Corollary
  \ref{adap.corollary}, assume
$$ \min_k |\beta_{0,k} |> \frac {9 ( \lambda_{adap}/\delta_{init}  + T  \lambda_0 )^2 c_0^2 \kappa^2 s_0 }{
2( \lambda_{adap}/\delta_{adap}  - T  \lambda_0 )}.$$ Then, under the
assumptions of Corollary \ref{adap.corollary}, on the set $\mathcal{J}$, $$S_0 \subset
\hat{S}_{adap} = \{1 \le k \le p : \hat{\beta}_{adap,k} \ne 0 \}$$
\end{enumerate}
\end{corollary}
The proof of Corollary \ref{screen.corollary} is given in the Appendix B.

\begin{center}
\section{Computational algorithm} \label{Sec5}
\end{center}
The algorithm for the estimators in (\ref{opt2}) and (\ref{opt3}) are based on the Block
Coordinate Gradient Descent (BCGD) method from \cite{Tseng07}.\\
The main ideas of our BCGD algorithm are that we cycle through the
coordinates and minimize the objective function $Q_{\lambda}(.)$ with respect
to only one coordinate while keeping the other parameters fixed (i.e. a 
Gauss-Seidel algorithm). In each
such step,  we approximate  $Q_{\lambda}(.)$ by a strictly convex quadratic
function. Then, we calculate a descent direction and we employ an inexact
line search to ensure a decrease in the objective function.\\
BCGD algorithms are used in \cite{Meie08} for the group
Lasso as well as in \cite{Wu08} and \cite{Fried08} for the ordinary
Lasso. We remark that \cite{Meie08} have a block structure
due to the grouped variables whereas we only focus on ungrouped
covariates. Thus the word "block" has no meaning in our context and
consequently, we omit it in the subsequent discussion. Furthermore, the ordinary Lasso has only regression parameters
to cycle through in contrast to our problem involving two kinds of parameters:
fixed regression and variance-covariance parameters.\\
Let us first introduce the notation and give an overview of the algorithm
before proving that our optimization problem achieves numerical
convergence. All the details as well as some computational aspects are
deferred to the Appendix C in the Supporting Information.\\
Let $\bm\phi^T=(\bm\beta^T,\bm\eta^T)\in \mathbb{R}^{p+q^*+1}$ be the
parametrization introduced in the previous section. Define the functions
\begin{equation*}
P(\bm{\phi}) := \sum_{k=1}^p |\beta_k| \quad , \quad
g(\bm{\phi}) := \frac{1}{2} \log|\bm{V}(\bm\eta)|
+ \frac{1}{2}(\bm{y}-\bm{X}\bm{\beta})^T\bm{V}(\bm\eta)^{-1}(\bm{y}-\bm{X}\bm{\beta}).\\
\end{equation*}
Now (\ref{opt2}) can be written as $\hat{\bm{\phi}}_{\lambda} =
\argmin_{\bm{\phi}} Q_{\lambda}(\bm{\phi}) := g(\bm{\phi}) +
\lambda P(\bm{\phi})$. Letting $\bm{e}_j$ the $j$th unit vector, the algorithm can be summarized in the following way:\\
 
\begin{Alg} {\bf(Coordinate Gradient Descent)} \label{Alg1}
\begin{itemize}
\item [(0)] Let $\bm{\phi}^0 \in \mathbb{R}^{p+q^*+1}$ be an initial value.
\end{itemize}
For  $\ell=0,1,2,\ldots$, let $\mathcal{S}^{\ell}$ be the index cycling through the
coordinates $\{1\}$, $\{2\}$,\ldots, $\{p+q^*\}$, $\{p+q^*+1\}$
\begin{itemize}
\item [(1)] Approximate the second derivative $\frac{\partial^2}{\partial
    (\phi_{\mathcal{S}^{\ell}})^2} Q_{\lambda}(\bm{\phi}^{\ell})$ by
  $h^{\ell} > 0$.
\item [(2)]  Calculate the descent direction\\
 $d^{\ell} := \argmin_{d \in \mathbb{R}}  \Big\{ g(\bm{\phi}^{\ell}) +
\frac{\partial}{\partial \phi_{\mathcal{S}^{\ell}}} g(\bm{\phi}^{\ell})d + 1/2 d^2 h^{\ell} +
  \lambda P(\bm{\phi}^{\ell} + d \bm{e}_{\mathcal{S}^{\ell}}) \Big\}$.
\item [(3)] Choose a stepsize $\alpha^{\ell} >0$  and set
  $\bm{\phi}^{\ell+1} = \bm{\phi}^{\ell} + \alpha^{\ell} d^{\ell}
  \bm{e}_{\mathcal{S}^{\ell}}$ such that there is a decrease in the objective function.
\end{itemize}
until convergence.
\end{Alg}

The details of (0) - (3) and further computational issues are
given in the Appendix C of the Supporting Information. An
implementation of the algorithm can be found in the \texttt{R} package
\texttt{lmmlasso}, which is available from the first author's website (\texttt{http://stat.ethz.ch/people/schell}) and will be made available on \texttt{CRAN}.\\

The convergence properties of the CGD algorithm are described in the following theorem.

\begin{theorem} {\bf(Convergence of the CGD algorithm)} \label{bcgthm}\\
If $(\bm{\phi}^{\ell})_{\ell \ge 0}$ is chosen according to Algorithm \ref{Alg1}, then every cluster point of $\{\bm{\phi}^{\ell}\}_{\ell \ge 0}$ is a stationary
point of $Q_{\lambda}(\bm{\phi})$.
\end{theorem}
The proof is given in the Appendix C.\\

In general, due to the non-convexity of the optimization problem, the CGD
algorithm may not achieve the global optimum.\\


\begin{center}
\section{Simulation study} \label{Sec6}
\end{center}
In this section, we assess the empirical performance of the
$\ell_1$-penalized maximum likelihood estimators (\ref{opt2}) and (\ref{opt3}) in different
kinds of simulation examples. We study several performance measures and compare the
proposed method with Lasso and linear mixed-effects methods, if
possible.\\
After some introductory remarks, we focus on high-dimensional
examples. The simulation study
for the low-dimensional setting is provided in the Supporting
Information. The application of the new procedure on a real data set is
illustrated in the next section.\\

Hereafter, we denote by \textit{lmmLasso} the $\ell_1$-penalized
maximum likelihood estimator (\ref{opt2}), by \textit{lmmadLasso} the
adaptive $\ell_1$-penalized
maximum likelihood estimator (\ref{opt3}) and by \textit{lme} the classical linear mixed-effects model
provided by the \texttt{R} package \texttt{nlme} \citep{PinJB2000}. Furthermore, let \textit{Lasso} denote the
standard Lasso \citep{Efro03} and \textit{adLasso} the adaptive Lasso
\citep{Zou06} where the regularization parameter is chosen by minimizing
the Bayesian Information Criterion.\\

As an overview, let us summarize the most important conclusions from
the simulation studies:
\begin{itemize}
\item [(a)] The variability of the estimated fixed-effects parameters
  $\hat{\beta}_k$ is much smaller if there is no corresponding random effect
  $(\hat{b}_i)_k$ for $i=1,\ldots,N$, for all \textit{lme}, \textit{lmmLasso} and \textit{lmmadLasso}.
\item [(b)] In the high-dimensional framework, the following aspects appear
  (and are virtually not observable in the low-dimensional setting):
\begin{enumerate}
\item Penalizing fixed-effects covariates which also incorporate a random
  effect causes bias problems. To be more specific, let us assume that the penalized $k$th
covariate has a fixed and a random-effects coefficient, i.e. $\beta_k$ and
$(b_i)_k$, respectively. If the regularization parameter $\lambda$ is large
and $\beta_k$ subject to penalization, then $\hat{\beta}_k$ is
shrunken towards zero. Thereby, the estimate of the
corresponding variance parameter gets large and $(\hat{b}_i)_k$ has a bias
related to the amount of shrinkage in $\hat{\beta}_k$. As a consequence,
covariates with fixed and random effect should no be subject to
  penalization.
\item An adaptive procedure (\ref{opt3}) with appropriate weights may
  reduce this adverse effect, but it does not overcome the aforementioned
  problem completely. The work of \cite{Bondell10} covers only the
  low-dimensional case and the authors do not present any
  parameter estimates in the simulation study.
\end{enumerate}
\item [(c)] There is a remarkable reduction of the estimated error variance
  $\hat{\sigma}^2$ when incorporating the random-effects structure in
  \textit{lmmLasso}, \textit{lmmadLasso}  and \textit{lme} compared with \textit{Lasso} and
  \textit{adLasso}. 
\item [(d)] The variability of the \textit{Lasso} and
  \textit{adLasso} coefficient estimators are larger than the corresponding variability of the mixed-effects model approaches. 
\item [(e)] If we focus on the identification of random-effects covariates,
  we suggest using a diagonal structure for $\bm\Psi$ and then eliminating
  those random-effects covariates with a small variance. An elaborate discussion of the
  selection of the random-effects structure is beyond the scope of this paper.
 In Section \ref{Sec7} we suggest a strategy how to remedy this problem.
\end{itemize}
In all subsequent simulation schemes, we restrict ourselves to the case where
all groups have the same number of observations, i.e. we set $n_i \equiv n$
for $i=1,\ldots,N$. Let the first column of $\bm{X}_i$ be the (non-penalized)
intercept. We assign  $\bm{Z}_i
\subset \bm{X}_i$ such that the columns of $\bm{Z}_i$ correspond to the
first $q$ columns of $\bm{X}_i$. This means that the first $q$ variables have
both a fixed-effects coefficient $\beta_k$ and a random-effects coefficient $(b_i)_k$ for
$i=1,\ldots,N$ and $k=1,\ldots,q$. The covariates are generated from a
multivariate normal distribution with mean zero and covariance matrix
$\bm{\Sigma}$ with the pairwise correlation
$\bm{\Sigma}_{kk'}=\rho^{|k-k'|}$  and $\rho=0.2$. Denote by $\bm\beta_0$
the true fixed effects and by $s_0:= \# \{1 \le k \le p : \beta_{0,k} \ne
0\}$ the true number of non-zero coefficients. Unless otherwise stated, we set $\bm\Psi=\theta^2\bm{I}$. In all subsequent tables, a non-penalized fixed-effects coefficient is marked by an asterisk $^*$.

\normalsize
\subsection{High-dimensional setting}
We study four examples in the high-dimensional setting ($\beta_{0,1}=1$ is
the unpenalized intercept).
\begin{itemize}
\item [$H_1$:] $N=25$, $n=6$, $N_T=150$, $p=300$, $q=2$,
  $\sigma^2=0.25$, $\theta^2=0.56$  and $s_0=5$ with $\bm\beta_0=(1,2,4,3,3,0,\ldots,0)^T$.
\item [$H_2$:] $N=30$, $n=6$, $N_T=180$, $p=500$, $q=1$,
  $\sigma^2=0.25$, $\theta^2=0.56$  and $s_0=5$ with $\bm\beta_0=(1,2,4,3,3,0,\ldots,0)^T$.
\item [$H_3$:] $N=30$, $n=6$, $N_T=180$, $p=1000$, $q=3$,
  $\sigma^2=0.25$, $\theta^2=0.56$  and $s_0=5$ with
  $\bm\beta_0=(1,2,4,3,3,0,\ldots,0)^T$.
\item [$H_4$:] $N=25$, $n=6$, $N_T=150$, $p=300$, $\sigma^2=0.25$,
\begin{equation*}
\bm\Psi=\begin{pmatrix} 3 &0 & 0 \\ 0 & 3& 0 \\ 0 & 0 &  2 \end{pmatrix}
\end{equation*}
and $s_0=5$ with $\bm\beta_0=(1,2,4,3,3,0,\ldots,0)^T$. In contrast to
the previous examples, we fit a wrong model assuming that $\bm\Psi$ is
diagonal with dimension 4.
\end{itemize}
The results in the form of means and standard deviations (in parentheses)
over 100 simulation runs are depicted in Table \ref{table4}, \ref{table5} and
\ref{table6}. Therein, $|S(\hat{\bm{\beta}})|$ denotes the cardinality of
the estimated active set and TP is the number of true positives.

\begin{table}[!h]
\scriptsize
\begin{center}
\caption{\textit{Simulation results for $H_1$, $H_2$ and $H_3$}} \label{table4}
\vspace{0.2cm}
\begin{tabular}{cc||cc|cc|ccccc}
\hline \hline
Model & Method& $|S(\hat{\bm\beta})|$ & TP & $\hat{\sigma}^2$ & $\hat{\theta}^2$
& $\hat{\beta}_1$ & $\hat{\beta}_{2}$ &  $\hat{\beta}_{3}$ & $\hat{\beta}_{4}$ & $\hat{\beta}_{5}$ \\ 
\hline
&     true& 5 & 5  & 0.25 &  0.56 & 1 & 2 & 4 & 3 & 3 \\
\hline
&lmmLasso&6.70  & 5 & 0.29 & 0.52& $1.01^*$ &$2.05^*$ & 3.86 & 2.90 & 2.88 \\
&&(2.14)&(0)&(0.05)&(0.12)&(0.16)&(0.16)& (0.06) & (0.06)& (0.06)\\
&lmmadLasso&6.59  & 5 & 0.22 & 0.52& $1.01^*$ &$2.03^*$ & 3.98 & 2.99 & 3.00 \\
$H_1$&&(2.02)&(0)&(0.04)&(0.12)&(0.16)&(0.16)& (0.06) & (0.05)& (0.05)\\
&Lasso & 6.29 & 5 &1.36& - &$1.01^*$& $2.07^*$& 3.76 & 2.84 & 2.79 \\
&       &(1.46)&(0)&(0.27)& - &(0.17)&(0.19)& (0.10) & (0.11)& (0.10)\\
&adLasso& 6.29 & 5 &1.16& - &$1.01^*$& $2.02^*$ & 3.98 & 3.00 & 2.99 \\
&       &(1.46)&(0)&(0.24)&-&(0.17)&(0.18)& (0.10) & (0.11)& (0.10)\\
\hline
&lmmLasso & 6.65 & 5 & 0.28 & 0.56 & $1.00^*$ & 1.90 & 3.91 & 2.92 & 2.89   \\
&   &(1.71)&(0)&(0.04)&(0.17)&(0.15)&(0.04)&(0.05)&(0.04)& (0.05)\\
&lmmadLasso & 6.53 & 5 & 0.22 & 0.55 & $1.00^*$ & 2.00 & 3.99 & 3.00 & 2.99   \\
$H_2$&   &(1.64)&(0)&(0.03)&(0.17)&(0.15)&(0.04)&(0.04)&(0.04)& (0.04)\\
&Lasso&6.84& 5 & 0.87 & - &$1.00^*$& 1.84 & 3.88 & 2.88 & 2.83 \\
&       &(2.02)&(0)&(0.19)& - &(0.15)&(0.08)& (0.07) & (0.09)& (0.08)\\
&adLasso&6.84& 5 & 0.72 & - &$1.00^*$& 2.00 & 4.00 & 3.00 & 2.98 \\
&       &(2.02)&(0)&(0.17)&-&(0.15)&(0.07)& (0.07) & (0.08)& (0.08)\\
\hline
&lmmLasso& 6.17 & 5 & 0.29 & 0.52 & $1.02^*$ &$2.00^*$&$4.04^*$& 2.84 & 2.84 \\
&   &(1.74)&(0)&(0.05)&(0.10)&(0.15)    &(0.15)& (0.15) & (0.07)& (0.06)\\
&lmmadLasso& 6.12 & 5 & 0.23 & 0.53 & $1.02^*$ &$2.00^*$&$4.00^*$& 2.99 & 2.99 \\
&   &(1.70)&(0)&(0.04)&(0.10)&(0.15)    &(0.15)& (0.15) & (0.07)& (0.06)\\
$H_3$&Lasso&5.93& 5 &1.94& - &$1.03^*$& $2.02^*$ & $4.06^*$ & 2.70 & 2.70 \\
&       &(1.48)&(0)&(0.36)&-&(0.17)  &(0.18)& (0.19) & (0.11)& (0.13)\\
&adLasso&5.93& 5 &1.69 & - &$1.03^*$& $2.02^*$ & $3.99^*$ & 2.98 & 2.97 \\
&       &(1.48)&(0)&(0.32)&-&(0.16)&(0.17)& (0.18) & (0.12)& (0.12)\\
\hline
\end{tabular}
\vspace{0.1cm}

 * indicates that the corresponding fixed-effects coefficient is not subject to penalization
\end{center}
\end{table}

\begin{table}[!h]
\scriptsize
\begin{center}
\caption{\textit{Simulation results for $H_4$}} \label{table5}
\vspace{0.2cm}
\begin{tabular}{c|cc|c|ccccc}
\hline
Method & $|S(\hat{\bm\beta})|$ & TP & $\hat{\sigma}^2$ 
& $\hat{\beta}_1$ & $\hat{\beta}_{2}$ &  $\hat{\beta}_{3}$ &
$\hat{\beta}_{4}$ & $\hat{\beta}_{5}$\\ 
\hline
 true& 5 & 5  & 0.25 & 1 & 2 & 4 & 3 & 3\\
\hline
lmmLasso&5.56 & 5 & 0.26 &$0.95^*$&$1.99^*$ &$3.97^*$ &$3.04^*$ & 2.82 \\
&(0.97)&(0)&(0.05)&(0.31)&(0.38)&(0.31)& (0.07) & (0.07)\\
lmmadLasso&5.56 & 5 & 0.22 &$0.95^*$&$1.99^*$ &$3.97^*$ &$3.00^*$ & 3.00 \\
&(0.97)&(0)&(0.04)&(0.31)&(0.38)&(0.30)& (0.07) & (0.06)\\
Lasso&6.84 & 5 & 7.85 &$0.94^*$&$2.01^*$ &$3.99^*$ &$3.11^*$ & 2.36\\
&(12.18)&(0)&(1.81)& (0.38)&(0.47)&(0.38)& (0.23) &(0.28)\\
adLasso&6.79 & 5 & 7.25 &$0.95^*$&$2.02^*$ &$4^*$ &$2.98^*$ & 3.01\\
&(11.68)&(0)&(1.76)&(0.37)&(0.47)&(0.38)& (0.22) & (0.29)\\
\hline
\end{tabular}
\vspace{0.1cm}

\scriptsize
* indicates that the corresponding fixed-effects coefficient is not subject to penalization
\end{center}
\end{table}

\begin{table}[!h]
\scriptsize
\begin{center}
\caption{\textit{Mean covariance estimates for $H_4$}} \label{table6}
\vspace{0.2cm}
\begin{tabular}{c|cccc}
\hline
Method & $\Psi_{11}$&$\Psi_{22}$&$\Psi_{33}$&$\Psi_{44}$\\ 
\hline
 true& 3 & 3 & 2 & 0\\
\hline
lmmLasso& 2.82 & 2.94 & 1.85 & 0.01 \\
        &(0.80)&(0.88)&(0.62)&(0.02)\\
lmmadLasso& 2.81 & 2.94 & 1.84 &0.01\\
          &(0.79)&(0.88)&(0.62)&(0.02)\\
\hline
\end{tabular}
\vspace{0.3cm}
\end{center}
\end{table}

Let us sum up the simulation results for the models $H_1$-$H_4$. As in
the low-dimensional setting (see Appendix D), the estimated active set is
sparse and all methods include the true non-zero coefficients.\\
Table \ref{table4} reveals that \textit{lmmLasso} and \textit{lmmadLasso} reduce
the error variance remarkably in comparison with \textit{Lasso} and
\textit{adLasso}. Nevertheless, \textit{lmmLasso} overestimates the true
value of $\sigma^2$ whereas \textit{lmmadLasso} underestimates
$\sigma^2$. We observe, in particular for $H_4$, that a maximum likelihood
approach (in contrast to a restricted maximum likelihood approach) gives biased variance-covariance estimators. It
is possible to implement a REML-type approach \citep{Ni10} in the high-dimensional
setting in order to reduce the bias in the variance
parameters. However, we have observed that $i)$ the number of
(Gauss-Seidel) cycles increases and $ii)$ the algorithm may fail to
converge.\\
In all models, we do not penalize the covariates with both a fixed
and random effect. Without doing this (not shown here), the fixed effects would
be set to zero whereas the estimated
between-subject variability $\hat{\theta}^2$ would increase. As a
consequence, the predicted random effects are too large and are not
centered at zero, but around the true fixed effect. Hence this would result
in a model which does not fulfill the assumptions in (\ref{model2})
anymore.\\
Table \ref{table4} and \ref{table5} reveal that the
variability of the fixed effects with no corresponding random effect is
approximately half of the non-penalized coefficients. This difference of
estimation variability is also observed
in the classical linear mixed-effects framework (see \textit{lme} in Table
7 and 8 in the Appendix D). 
Besides, \textit{lmmLasso} has a bias towards zero, which
is notably smaller than that from the \textit{Lasso}. As expected, this
bias can be reduced by \textit{lmmadLasso}.\\
Concerning $H_4$, it is worth to point out that although not knowing the true
covariance structure, we may use a diagonal structure for $\bm\Psi$ and then
drop the variances which are close to zero. A suggestion how to use this idea in a real data set is presented in the next section.

\subsection{Within-group prediction performance}
We now turn to consider the performance of the proposed methodology concerning
within-group prediction. We compare the predictive performance between six
different Lasso procedures. In doing so, denote by \textit{lmmLasso},
\textit{lmmadLasso}, \textit{Lasso} and \textit{adLasso} the procedures
from the previous subsection. In addition, let \textit{cv-Lasso} be a
cross-validated Lasso and \textit{cv-adLasso} a cross-validated adaptive
Lasso whose $\lambda$-value is chosen by $10$-fold cross-validation.\\
We fix the following scenario: $N=25$, $n_i \equiv6$ for $i=1,\ldots,N$,
$q=3$, $s_0=5$ with $\bm\beta_0=(1,1.5,1.2,1,2,0,\ldots,0)^T$, $\sigma^2=1$ and
$\rho=0.2$. We only alter the number of fixed
covariates $p$ and the variance component $\theta^2$. For measuring the
quality of prediction, we generate a test set with $50$ observations per group and calculate the mean squared prediction error.
The three models considered are 
\begin{center}
$P_1$: $p=10$, $P_2$: $p=100$
and $P_3$: $p=500$.
\end{center}
The results are shown in Table \ref{table7}.
\linespread{1} 
\begin{table}[h]
\scriptsize
\begin{center}
\caption{\textit{Mean squared prediction error for three simulation examples.}} \label{table7}
\vspace{0.3cm}
\begin{tabular}{cc|cccccc}
\hline \hline
\tiny{Model} & \tiny{$\theta^2$} & \tiny{lmmLasso} & \tiny{lmmadLasso} &\tiny{Lasso} & \tiny{adLasso} & \tiny{cv-Lasso} & \tiny{cv-adLasso} \\ 
\hline
$P_1$& 0 &1.01&1.02&1.00&1.01&1.05&1.01 \\
$(p=10)$ & 0.25 &1.33&1.29&1.76&1.84&1.81&1.84 \\
& 1 &1.66&1.55&3.74&3.74&3.88&3.77 \\
& 2 &1.67&1.80&5.92&6.25&5.94&6.25 \\
\hline
 $P_2$& 0 &1.12&1.02&1.26&1.09&1.20&1.14 \\
$(p=100)$& 0.25 &1.51&1.38&1.75&1.75&2.06&1.75 \\
& 1 &1.94&1.86&4.35&4.53&4.61&4.23 \\
& 2 &2.49&1.95&7.04&7.02&7.09&6.98 \\
\hline
$P_3$ & 0 &1.22&1.07&1.18&1.26&1.24&1.58 \\
$(p=500)$& 0.25 &1.83&1.58&2.63&2.67&2.98&3.58 \\
& 1 &2.00&1.85&4.35&3.78&4.14&4.85 \\
& 2 &2.54&2.04&10.30&8.26&9.47&11.28 \\
\hline
\end{tabular}\\
\vspace{0.3cm}
\end{center}
\end{table}

\linespread{1}
We see that the methods differ slightly for $\theta^2=0$  which corresponds
to no grouping structure. As $\theta^2$ increases, the mean squared
prediction error increases less for the \textit{lmmLasso} and the \textit{lmmadLasso} than for the other methods. These results highlight
that we can indeed achieve prediction improvements using the suggested
mixed-effects model approach if the underlying model is given by (\ref{model2}).

\section{Application: riboflavin data} \label{Sec7}
\textit{Data description.} We illustrate the proposed procedure on a real data set which is provided by
DSM (Switzerland). The response variable is the logarithm of the riboflavin
production rate of Bacillus subtilis. There are $p=4088$ covariates
measuring the gene expression levels. We have $N=28$ groups with
$n_i \in \{2,\ldots,6\}$ and $N_T=111$ observations. We standardize
all covariates to have mean zero and variance one.\\

\noindent\textit{Model selection strategy.} Preliminary, we address the issue of determining those covariates which have both
a fixed and a random-effects coefficient. In other words, we
specify the matrices $\bm{Z}_i \subset \bm{X}_i$. Since we have to deal with
high-dimensional, low sample size data, the various tools proposed
in \cite{PinJB2000} for determining $\bm{Z}_i$ can hardly be applied. Instead,
we suggest the following strategy:
\begin{enumerate}
\item [(1)] Calculate an ordinary Lasso solution $\hat{\bm\beta}^{Lasso}$
  (with cross-validation)
  and define the active set $\hat{S}_{init}:=\{1 \le k \le p : \hat{\beta}_k^{Lasso} \ne 0
  \}$.
\item [(2)] For each $l \in \hat{S}_{init}$, fit a model in which only the $l$th
  covariate has
  a random-effects coefficient. Denote the corresponding variance estimate by $\hat{\theta}^2_l$.
\item [(3)] Let $\hat{\theta}^2_{[1]}\ge\hat{\theta}^2_{[2]}\ge\ \ldots \ge
  \hat{\theta}^2_{[|S_{init}|]}$ be the ordered estimated variances from
  (2). Then for $\kappa>0$
  define the set $\mathcal{R}_{\kappa}:=\{l \in S_{init}: \hat{\theta}_l^2
  > \kappa \} \cap \{l \in \hat{S}_{init}: BIC_{\hat{\theta}_l^2} \le
  BIC_{0} \}$ where $BIC_{0}$ is the
  BIC of the Lasso solution in (1).
\item [(4)] Fit a model with
  $\bm{Z}_i=\bm{X}_i^{\mathcal{R}_{\kappa}}$ (where
  $\bm{X}_i^{\mathcal{R}_{\kappa}}$ consists of the variables included in
  $\mathcal{R}_{\kappa}$) and $\bm\Psi$ being diagonal and keep the
  non-zero elements of $\hat{\bm\Psi}$.
\end{enumerate}
By doing so (and setting $\kappa=0.05)$, it seems reasonable to fit a model
wherein two covariates have an additional random effect. Denoting these
variables as $k_1$ and $k_2$, the model can
be written as
\begin{equation} \label{ranintmod}
y_{ij} = \bm{x}_{ij}^T\bm{\beta} + b_{ik_1}z_{ijk_1} + b_{ik_2}z_{ijk_2} + \varepsilon_{ij} \quad i=1,\ldots,N,
\quad j=1,\ldots,n_i
\end{equation}
with $b_{ik_1} \sim \mathcal{N}(0,\theta_{k_1}^2)$, $b_{ik_2} \sim
\mathcal{N}(0,\theta_{k_2}^2)$  and $\varepsilon_{ij}\sim \mathcal{N}(0,\sigma^2)$.\\

\noindent\textit{Results.} We compare the results of \textit{lmmLasso} and
\textit{lmmadLasso} with \textit{Lasso} and \textit{adLasso}. The variance
component estimates, the cardinality of the active set and the rank $R$ of
five fixed-effects coefficients are shown in Table \ref{table8}. The
ranking is determined by ordering the absolute values of the fixed-effects coefficients.

\begin{table}[!h]
\footnotesize
\begin{center}
\caption{\textit{Results for \textit{lmmLasso}, \textit{lmmadLasso},
    \textit{Lasso} and \textit{adLasso} of the riboflavin data set}} \label{table8}
\vspace{0.3cm}
\begin{tabular}{c|cccc}
\hline
 &\tiny{lmmLasso} &\tiny{lmmadLasso} & \tiny{Lasso} & \tiny{adLasso} \\ 
\hline
$\hat{\sigma}^2$       & 0.18 & 0.15 &0.30 &0.20 \\
$\hat{\theta}_{k_1}^2$  & 0.17 & 0.08 &- &- \\
$\hat{\theta}_{k_2}^2$  & 0.03 & 0.03 &- &- \\
$|S(\hat{\bm\beta})|$  & 18    & 14 & 21 & 20 \\
$R_{\hat{\beta}_{1}}$   & 1 & 1 & 1 & 1 \\ 
$R_{\hat{\beta}_{2}}$   & 2 & 2 & 4 & 6 \\ 
$R_{\hat{\beta}_{3}}$   & 3 & 3 & 3 & 5 \\ 
$R_{\hat{\beta}_{4}}$   & 4 & 13 & - & - \\ 
$R_{\hat{\beta}_{5}}$   & 5 & 6 & 6 & 7 \\ 
\hline
\end{tabular}\\
\vspace{0.3cm}
\end{center}
\end{table}
From Table \ref{table8}, we see that the error variance of the
\textit{Lasso} may be considerably reduced by the \textit{lmmLasso}. Although the
variance $\hat{\theta}_{k_2}^2$ is small, the BIC of this model is smaller
than that of the model including only $k_1$ as random covariate. It is noteworthy
that $53\%$ of the total variability in the data set is due to the
between-group variability. This strongly indicates that there is indeed
some variability between the groups. As might have been expected from the
simulation results, the active set of \textit{lmmLasso} and
\textit{lmmadLasso} is smaller than the active set from \textit{Lasso} and
\textit{adLasso}. The ranking indicates that there is one
dominating covariate whereas the other coefficients differ only slightly
between the four procedures (not shown).  




\section{Discussion}
We present an $\ell_1$-penalized maximum likelihood estimator
for high-dimensional linear mixed-effects models. The proposed methodology
copes with the difficulty of combining a non-convex loss function and an
$\ell_1$-penalty. Thereby, we deal with
theoretical and computational aspects which are substantially more challenging than in the
linear regression setting. We prove theoretical results concerning the consistency of the estimator
and we present a non-asymptotic oracle result for the adaptive
$\ell_1$-penalized estimator. Moreover, by developing a coordinate
gradient descent algorithm, we achieve provable numerical convergence of our
algorithm to at least a stationary point. 
Our simulation studies and real data example show that the error variance
can be remarkably reduced when incorporating the knowledge
about the cluster structure among observations.\\ \vspace{0.3cm}

\noindent\textbf{\Large Supporting Information} \vspace{0.4cm}

\noindent Additional Supporting Information can be found in the Appendices
of this article:\\

\noindent \textbf{Appendix A.} Proof of Theorem 1 from Section \ref{Sec3}.\\
\textbf{Appendix B.} Proof of Theorem 2 from Section \ref{Sec3}.\\
\textbf{Appendix C.} Computational details of Algorithm 1 in Section
\ref{Sec5}.\\
\textbf{Appendix D.} Simulations for the low-dimensional setting in
Section \ref{Sec6}.\\

\bibliographystyle{/sfs/u/staff/schell/Diss/latex/paper/JASA/ECA_jasa}
\bibliography{/u/schell/Diss/latex/myLit}

\begin{thebibliography}{}

\bibitem[Bertsekas, 1999]{BerD99}
Bertsekas, D.~P. (1999).
\newblock {\em Nonlinear Programming}.
\newblock Athena Scientific, Belmont.

\bibitem[Bickel et~al., 2009]{Bickel09}
Bickel, P., Ritov, Y., and Tsybakov, A. (2009).
\newblock Simultaneous analysis of lasso and dantzig selector.
\newblock {\em The Annals of Statistics}, 37:1705--1732.

\bibitem[Bondell et~al., 2010]{Bondell10}
Bondell, H.~D., Krishna, A., and Ghosh, S.~K. (2010).
\newblock Joint variable selection of fixed and random effects in linear
  mixed-effects models.
\newblock {\em Biometrics}, In Press.

\bibitem[B\"uhlmann and van~de Geer, 2011]{BuhlGeer10}
B\"uhlmann, P. and van~de Geer, S. (2011).
\newblock {\em Statistics for High-Dimensional Data: Methods, Theory and
  Applications}.
\newblock Springer.

\bibitem[Bunea et~al., 2007]{Bunea07}
Bunea, F., Tsybakov, A., and Wegkamp, M. (2007).
\newblock Sparsity oracle inequalities for the lasso.
\newblock {\em Electronic Journal of Statistics}, 1:169--194.

\bibitem[Candes and Tao, 2007]{Candes07}
Candes, E. and Tao, T. (2007).
\newblock The dantzig selector: Statistical estimation when $p$ is much larger
  than $n$.
\newblock {\em The Annals of Statistics}, 35:2313--2351.

\bibitem[Demidenko, 2004]{Demi04}
Demidenko, E. (2004).
\newblock {\em Mixed Models, Theory and Applications}.
\newblock Wiley.

\bibitem[Efron et~al., 2004]{Efro03}
Efron, B., Hastie, T., Johnstone, I., and Tibshirani, R. (2004).
\newblock Least angle regression.
\newblock {\em The Annals of Statistics}, 32:407--499.

\bibitem[Friedman et~al., 2010]{Fried08}
Friedman, J., Hastie, T., and Tibshirani, R. (2010).
\newblock Regularization paths for generalized linear models via coordinate
  descent.
\newblock {\em Journal of Statistical Software}, 33.

\bibitem[Greenshtein and Ritov, 2004]{Greenshtein04}
Greenshtein, E. and Ritov, Y. (2004).
\newblock Persistence in high-dimensional linear predictor selection and the
  virtue of overparametrization.
\newblock {\em Bernoulli}, 10:971--988.

\bibitem[Huang et~al., 2008]{HuangMaZhang08}
Huang, J., Ma, S., and Zhang, C.-H. (2008).
\newblock Adaptive lasso for sparse high-dimensional regression models.
\newblock {\em Statistica Sinica}, 18:1603--1618.

\bibitem[Ibrahim et~al., 2010]{Ibrahim10}
Ibrahim, J.~G., Zhu, H., Garcia, R.~I., and Guo, R. (2010).
\newblock Fixed and random effects selection in mixed effects models.
\newblock {\em Biometrics}, In Press.

\bibitem[Laird and Ware, 1982]{Lair82}
Laird, N.~M. and Ware, J.~H. (1982).
\newblock Random-effects models for longitudinal data.
\newblock {\em Biometrics}, 83:1014--1022.

\bibitem[Liu et~al., 2008]{Liu08}
Liu, H., Tang, Y., and Zhang, H.~H. (2008).
\newblock A new chi-square approximation to the distribution of non-negative
  definite quadratic forms in non-central normal variables.
\newblock {\em Computational Statistics and Data Analysis}, 53:853--856.

\bibitem[McCulloch and Searle, 2001]{McCCS01}
McCulloch, C.~E. and Searle, S.~R. (2001).
\newblock {\em Generalized, Linear, and Mixed Models}.
\newblock Wiley Series in Probability and Statistics. Wiley.

\bibitem[Meier et~al., 2008]{Meie08}
Meier, L., van~de Geer, S., and B\"uhlmann, P. (2008).
\newblock The group lasso for logistic regression.
\newblock {\em Journal of the Royal Statistical Society: Series B}, 70:53--71.

\bibitem[Meinshausen and B\"uhlmann, 2006]{Mein06}
Meinshausen, N. and B\"uhlmann, P. (2006).
\newblock High-dimensional graphs and variable selection with the lasso.
\newblock {\em The Annals of Statistics}, 34:1436--1462.

\bibitem[Meinshausen and B\"uhlmann, 2010]{MeinBuhl10}
Meinshausen, N. and B\"uhlmann, P. (2010).
\newblock Stability selection (with discussion).
\newblock {\em Journal of the Royal Statistical Society: Series B},
  72:417--473.

\bibitem[Meinshausen et~al., 2009]{MeinMeiBuhl09}
Meinshausen, N., Meier, L., and B\"uhlmann, P. (2009).
\newblock p-values for high-dimensional regression.
\newblock {\em Journal of the American Statistical Association},
  104:1671--1681.

\bibitem[Meinshausen and Yu, 2009]{Meins09}
Meinshausen, N. and Yu, B. (2009).
\newblock Lasso-type recovery of sparse representations for high-dimensional
  data.
\newblock {\em The Annals of Statistics}, 37:246--270.

\bibitem[Ni et~al., 2010]{Ni10}
Ni, X., Zhang, D., and Zhang, H.~H. (2010).
\newblock Variable selection for semiparametric mixed models in longitudinal
  studies.
\newblock {\em Biometrics}, 66:79--88.

\bibitem[Osborne et~al., 2000]{OsbPresTur00}
Osborne, M.~R., Presnell, B., and Turlach, B. (2000).
\newblock A new approach to variable selection in least squares problems.
\newblock {\em IMA Journal of Numerical Analysis}, 20(3):389--403.

\bibitem[Pinheiro and Bates, 1996]{Pinh96}
Pinheiro, J. and Bates, D. (1996).
\newblock Unconstrainted parametrizations for variance-covariance matrices.
\newblock {\em Statistics and Computing}, 6:289--296.

\bibitem[Pinheiro and Bates, 2000]{PinJB2000}
Pinheiro, J.~C. and Bates, D.~M. (2000).
\newblock {\em Mixed-Effects Models in S and S-Plus}.
\newblock Springer, New York.

\bibitem[St\"adler et~al., 2010]{Stad09}
St\"adler, N., B\"uhlmann, P., and van~de Geer, S. (2010).
\newblock $l_1$-penalization for mixture regression models (with discussion).
\newblock {\em Test}, 19:209--285.

\bibitem[Tibshirani, 1996]{Tibs96}
Tibshirani, R. (1996).
\newblock Regression shrinkage and selection via the lasso.
\newblock {\em Journal of the Royal Statistical Society: Series B},
  58:267--288.

\bibitem[Tseng and Yun, 2009]{Tseng07}
Tseng, P. and Yun, S. (2009).
\newblock A coordinate gradient descent method for nonsmooth separable
  minimization.
\newblock {\em Mathematical Programming: Series B}, 117:387--423.

\bibitem[van~de Geer, 2008]{Geer08}
van~de Geer, S. (2008).
\newblock High-dimensional generalized linear models and the lasso.
\newblock {\em The Annals of Statistics}, 36:614--645.

\bibitem[van~de Geer and B\"uhlmann, 2009]{GeerBuhl09}
van~de Geer, S. and B\"uhlmann, P. (2009).
\newblock On the conditions used to prove oracle results for the lasso.
\newblock {\em Electronic Journal of Statistics}, 3:1360--1392.

\bibitem[van~de Geer et~al., 2010]{GeerBuhlZhou10}
van~de Geer, S., B\"uhlmann, P., and Zhou, S. (2010).
\newblock The adaptive and the thresholded lasso for potentially misspecified
  models.
\newblock {\em Preprint arXiv:1001.5176v3}.

\bibitem[Verbeke and Molenberghs, 2000]{VerbMole00}
Verbeke, G. and Molenberghs, G. (2000).
\newblock {\em Linear Mixed Models for Longitudinal Data}.
\newblock Springer, New York.

\bibitem[Wasserman and Roeder, 2009]{WassRoed09}
Wasserman, L. and Roeder, K. (2009).
\newblock High-dimensional variable selection.
\newblock {\em Annals of Statistics}, 37:2178--2201.

\bibitem[Wu and Lange, 2008]{Wu08}
Wu, T. and Lange, K. (2008).
\newblock Coordinate descent algorithms for lasso penalized regression.
\newblock {\em The Annals of Applied Statistics}, 2:224--244.

\bibitem[Zhang and Huang, 2008]{Zhang08}
Zhang, C.-H. and Huang, J. (2008).
\newblock The sparsity and bias of the lasso selection in high-dimensional
  linear regression.
\newblock {\em The Annals of Statistics}, 36:1567--1594.

\bibitem[Zhao and Yu, 2006]{Zhao06}
Zhao, P. and Yu, B. (2006).
\newblock On model selection consistency of lasso.
\newblock {\em Journal of Machine Learing Research 7}, 7:2541--2563.

\bibitem[Zou, 2006]{Zou06}
Zou, H. (2006).
\newblock The adaptive lasso and its oracle properties.
\newblock {\em Journal of the American Statistical Association},
  101:1418--1429.

\bibitem[Zou et~al., 2007]{Zou07}
Zou, H., Hastie, T., and Tibshirani, R. (2007).
\newblock On the "degrees of freedom" of the lasso.
\newblock {\em The Annals of Statistics}, 35:2173--2192.

\end{thebibliography}

\vspace{0.8cm}
\noindent J\"urg Schelldorfer, Seminar f\"ur Statistik, Department of Mathematics, ETH Zurich, CH-8092 Zurich, Switzerland.\\
E-mail: schelldorfer@stat.math.ethz.ch

\newpage

\begin{center}
\emph{Supporting Information to the paper} \\\vspace{0.5cm}

\large{Estimation for High-Dimensional Linear Mixed-Effects Models Using
  $\ell_1$-Penalization}\\\vspace{0.5cm}
 
\normalsize
J\"urg Schelldorfer, Peter B\"uhlmann and \\Sara van de Geer
\end{center}

\section*{Appendix A: Proof of Theorem 1}
The proof consists of three parts. Firstly, we need an inequality ensuring that Lemma \ref{Lemma2} holds. Secondly, we show
that the probability (\ref{231209}) in Lemma \ref{Lemma2} is large. And for
completion of our proof, we can then refer to \cite{Stad09}.\\
From model (5), the log-likelihood function of
$\bm{y}_i$ with respect to the parametrization in (12) is given by
\footnotesize
\begin{equation*}
\ell_{\bm{\phi}}(\bm{y}_i):=-\frac{n}{2}\log(2\pi) -\frac{1}{2}\log|\bm{Z}_i\bm\Psi_{\theta}\bm{Z}_i^T+e^{\varrho}\bm{I}|-\frac{1}{2}(\bm{y}_i-\bm{X}_i\bm\beta)^T(\bm{Z}_i\bm\Psi_{\theta}\bm{Z}_i^T+e^{\varrho}\bm{I})^{-1}(\bm{y}_i-\bm{X}_i\bm\beta)
\end{equation*}
\normalsize
Then, define the score function $s_{\bm{\phi}}(\bm{y}_i):= \partial/\partial \bm{\phi} \ell_{\bm{\phi}}(\bm{y}_i)$.
\begin{lemma} \label{Lemma1} Under Assumption 1, there exist constants $c_1,c_2,c_3 \in \mathbb{R}_+$ such that
\begin{equation*}
\sup_{\bm{\phi} \in \bm{\Phi}}\|s_{\bm{\phi}}(\bm{y}_i)\|_{\infty} \le G_1(\bm{y}_i):=c_1 +
c_2\|\bm{y}_i\|_2 + c_3\|\bm{y}_i\|_2^2 \quad i=1,\ldots,N.
\end{equation*}
\end{lemma}
\begin{proof}
The proof is straightforward using on the one hand the Cauchy-Schwarz
inequality and the fact that the induced $L_2$-norm of a square 
matrix $\bm{A}$ is given by $\|\bm{A}\|_2 =\sqrt{\lambda_{max}(\bm{A}^T\bm{A})}$, where $\lambda_{max}$ denotes
the largest eigenvalue. On the other hand, we conclude from Assumption
1 and (12) that the eigenvalues of
$\bm{Z}\bm\Psi_{\bm\theta}\bm{Z}^T$ and
$\bm{Z}\partial/\partial\theta_j\bm\Psi_{\bm\theta}\bm{Z}^T$ are bounded.
\end{proof}

Now we introduce the empirical process and present a result which
controls the increments of it. The Lemma below gives a lower
bound for the probability that the increments are small. Afterwards, we show that
this lower bound is large.\\
Define the empirical process
\begin{equation*}
V_N(\bm{\phi}):= \frac{1}{N} \sum_{i=1}^N \Big\{\ell_{\bm{\phi}}(\bm{y}_i) -
  \mathbb{E}[\ell_{\bm{\phi}}(\bm{y}_i)]  \Big\}
\end{equation*}
and
\begin{equation} 
\lambda_0 = M_N \log N \sqrt{\frac{\log(p \vee N)}{N}} \tag{A.1}\label{12.08}  .
\end{equation}

\begin{lemma} \label{Lemma2} Under Assumption 1 and for
  constants $a_1$, $a_2$ and $a_3$ depending on $K$ and for all $T \ge a_1$,
\begin{equation*} 
\sup_{\bm{\phi} \in \bm{\Phi}} \frac{\Big|V_N(\bm{\phi})-V_N(\bm{\phi}_0)
  \Big|}{(\|\bm\beta-\bm\beta_0\|_1 + \|\bm{\eta}-\bm{\eta}_0\|_2)\vee \lambda_0} \le T \lambda_0
\end{equation*}
with probability at least
\begin{equation}
1-a_2\exp\Big[-\frac{T^2 \log^2N \log(p\vee N)}{a_3^2} \Big]
-\mathbb{P}\Bigg(\frac{1}{N} \sum_{i=1}^N F(\bm{y}_i) >
\frac{T\lambda_0^2}{dK} \Bigg)  \tag{A.2} \label{231209}
\end{equation}
where $d:=n+q^*+1$ and
\begin{equation} 
F(\bm{y}_i)=G_1(\bm{y}_i)\bm{1}_{\{G_1(\bm{y}_i) > M_N \}} + \mathbb{E}
\Big[G_1(\bm{y}_i)\bm{1}_{\{G_1(\bm{y}_i) > M_N\}} \Big]
\tag{A.3}\label{211109}  .
\end{equation}
\end{lemma} 
The proof of Lemma \ref{Lemma2} is given in \cite{Stad09}. Next, we show that the third term is small in our setting.
\begin{lemma} \label{Lemma3}
There are constants $b_1$ and $b_2$ depending on $K$ and $n$, a constant $\rho$
depending on $T$, $n$ and $K$ such that for any $0<\varepsilon<1/2$
and $M_N:=b_1(2 \sqrt{\log N}+\sqrt{b_2})^2$ we have
\begin{equation*}
\mathbb{P} \Bigg(\frac{1}{N}
\sum_{i=1}^N F(\bm{y}_i)>\frac{T\lambda_0^2}{dK} \Bigg)
\le \frac{\rho}{\log N}\frac{1}{N^{1-2\varepsilon}}  .
\end{equation*}
\end{lemma}
\begin{proof}
In the subsequent discussion, if $A$ is a constant, we
assume throughout that $N$ is large enough such that $M_N-A>0$. From (\ref{12.08}) we see that it suffices to show that for a constant $a_4$,
\begin{equation} 
\mathbb{P} \Bigg(\frac{1}{N}
\sum_{i=1}^N F(\bm{y}_i)>a_4 \frac{\log N}{N} \Bigg)
\le \frac{\rho}{\log N}\frac{1}{N^{1-2\varepsilon}}  . \tag{A.4}\label{251109}
\end{equation}
The expectation in (\ref{211109}) only affects the constants in the
remainder of the proof. Therefore, we omit this term in the sequel. From
\footnotesize
\begin{equation*}
\mathbb{P}[c_1+c_2\|\bm{y}_i\|_2 + c_3\|\bm{y}_i\|_2^2>M_N] \le
\mathbb{P}\Big[\|\bm{y}_i\|_2^2>\Big(\frac{M_N-c_1}{2c_2} \Big)^2\Big]
+ \mathbb{P}\Big[\|\bm{y}_i\|_2^2 > \frac{M_N-c_1}{2c_3}\Big]  ,
\end{equation*} 
\normalsize
and the fact that $M_N \to \infty $, we deduce that we can restrict ourselves to
the analysis of $\mathbb{P}[\|\bm{y}_i\|_2^2 > M_N]$. For the sake of notational simplicity, we will leave out the index $i$ and
show that for an appropriate definition of $M_N$, 
\begin{equation} 
\mathbb{P}[\|\bm{y}\|_2^2 > M_N] \le \frac{n}{N^2} \tag{A.5} \label{041009}
 .
\end{equation}

Denote by $\chi^2_{\nu}(\delta)$ the noncentral $\chi^2$ distribution with $\nu$
degrees of freedom and non-centrality parameter $\delta$. 
The following identity holds \citep{Liu08}.
\begin{claim} \label{claim1} If $\bm{y} \thicksim \mathcal{N}_n(\bm{\mu},\bm{V})$ with $\bm{\mu} \in
  \mathbb{R}^{n}$ and $\bm{V} \in \mathbb{R}^{n \times n}$ positive
  definite, then $
\|\bm{y}\|_2^2 = \bm{y}^T\bm{y} = \sum_{j=1}^n \lambda_j \chi^2_1(\delta_j)$
where $\{\chi^2_1(\delta_j)\}_{j=1}^n$ are independent, $\lambda_j$ for
$j=1,\ldots,n$ are the eigenvalues of $\bm{V}$ and if $\bm{V}=\bm{UDU}^T$ for an
orthonormal matrix $\bm{U}$, then $\delta_j=(\bm{U}^T\bm{{V}}^{-1/2}\bm{\mu})^2_j$.
\end{claim} 

\begin{claim} \label{claim2}
\begin{equation*}
\mathbb{P}[\chi^2_1(\delta)>M] \le
\frac{1}{\sqrt{M}-\sqrt{\delta}} \frac{2}{\sqrt{2 \pi}}
\exp\Bigg(-\frac{(\sqrt{M}-\sqrt{\delta})^2}{2}  \Bigg)   .
\end{equation*}
\end{claim}
\begin{proof} 
If $X \thicksim \mathcal{N}(\mu,\zeta^2)$, then by definition of the noncentral
 $\chi^2$ distribution $(X/\zeta)^2 \thicksim  \chi^2_{\nu=1}(\delta=(\mu/\zeta)^2)$. 
Hence $
 \mathbb{P}[\chi^2_1(\delta)>M] =  2 \cdot
 \mathbb{P}[\frac{X}{\zeta}>\sqrt{M}] 
   = 2 \cdot \mathbb{P}[\frac{X-\mu}{\zeta}>\sqrt{M}-\sqrt{\delta}]  = 2
   \cdot S(\sqrt{M}-\sqrt{\delta})$,  where $S(t):=\frac{1}{\sqrt{2\pi}} \int_t^{\infty}\exp(-u^2/2) du$ is the
 survival function of a standard Gaussian random variable for which the
 following inequalities hold:
 \begin{equation*}
 \frac{t}{1+t^2}\frac{1}{\sqrt{2 \pi}}\exp(-t^2/2) < S(t) < \frac{1}{t}
 \frac{1}{\sqrt{2\pi}}\exp(-t^2/2) \quad \textrm{for} \quad t > 0  .
 \end{equation*}
Thus, we conclude
\begin{equation*}
\mathbb{P}[\chi^2_1(\delta)>M] \le
\frac{1}{\sqrt{M}-\sqrt{\delta}} \frac{2}{\sqrt{2 \pi}}
\exp\Bigg(-\frac{(\sqrt{M}-\sqrt{\delta})^2}{2}  \Bigg)  .
\end{equation*}
\end{proof}
 
\begin{claim} \label{claim3} For $M_{N,\delta} := (2 \sqrt{\log N}+\sqrt{\delta})^2$, 
\begin{equation*}
\mathbb{P}[\chi^2_1(\delta)>M_{N,\delta}] \le \frac{1}{N^2}.
\end{equation*}
\end{claim}
\begin{proof}
Using Claim \ref{claim2},
\begin{align*}
\mathbb{P}[\chi^2_1(\delta)>M_{N,\delta}] & \le 
\frac{1}{\sqrt{M_{N,\delta}}-\sqrt{\delta}}\frac{2}{\sqrt{2\pi}}
\exp(-\frac{(\sqrt{M_{N,\delta}}-\sqrt{\delta})^2}{2}) \\
 & \le 1 \cdot \exp(-\frac{(2 \sqrt{\log N}+\sqrt{\delta}-\sqrt{\delta})^2}{2})
 \le \frac{1}{N^2} .
\end{align*}
\end{proof}

\begin{claim} \label{claim4} For the eigenvalue vector
  $\bm\lambda:=(\lambda_1,\ldots,\lambda_n)$ of $\bm{V}$, the
  non-centrality parameter vector $\bm\delta:=(\delta_1,\ldots,\delta_n)$,
  define
\begin{align*}
& \lambda_{max} := \max_{1 \le j \le n} \lambda_j \\
&M_{N,n,\bm\lambda,\delta_j}:=n\lambda_{max}(2 \sqrt{\log
  N}+\sqrt{\delta_j})^2 \\
&\delta:=\argmax_{\delta_j,1 \le j \le
  n}\mathbb{P}[\chi^2_1(\delta_j)>\frac{M_{N,n,\bm\lambda,\delta_j}}{n\lambda_{max}}]\\
&M_{N,n,\bm\lambda,\bm\delta}:=M_{N,n,\bm\lambda,\delta}  ,
\end{align*}
then
\begin{equation*}
 \mathbb{P}[\|\bm{y}\|_2^2>M_{N,n,\bm\lambda,\bm\delta}] \le \frac{n}{N^2}.
\end{equation*}
\end{claim}
\begin{proof} For any $M > 0$, using Claim \ref{claim1} and \ref{claim2}, 
\footnotesize
\begin{align*}
\mathbb{P}[\|\bm{y}\|_2^2>M]& =  \mathbb{P}\Big[\sum_{j=1}^n\lambda_j\chi^2_1(\delta_j)>M\Big]
\le \sum_{j=1}^n \mathbb{P}\Big[\chi^2_1(\delta_j)>\frac{M}{n\lambda_j}\Big]
 \le \sum_{j=1}^n
\mathbb{P}\Big[\chi^2_1(\delta_j)>\frac{M}{n\lambda_{max}}\Big]\\
&  \le n \cdot \max_{1 \le j \le n}
\mathbb{P}\Big[\chi^2_1(\delta_j)>\frac{M}{n\lambda_{max}}\Big].
\end{align*}
\normalsize
Set $M = M_{N,n,\bm\lambda,\bm\delta}$ and using Claim \ref{claim3}
\begin{align*}
\mathbb{P}[\|\bm{y}\|_2^2>M_{N,n,\bm\lambda,\bm\delta}]
\le n \cdot \mathbb{P}[\chi^2_1(\delta) > (2 \sqrt{\log
  N}+\sqrt{\delta})^2]
 \le \frac{n}{N^2}  .
\end{align*}
\end{proof}
At this point, we have proven (\ref{041009}). Due to Assumption 1 (and by using the same
techniques as in the proof of Lemma \ref{Lemma1}) , $\lambda_{max}^{(i)} \le \frac{1}{2}q^2(q+1)K^3:=b_1$ and $\delta_j^{(i)} \le
nK^2e^K:=b_2$ for all $i$ and $j$. Thereby, we define $$M_N:=b_1(2
\sqrt{\log N}+\sqrt{b_2})^2  .$$ Hence we choose $M_N$ of the order
$\log N$.
We now use these results to derive formula (\ref{251109}),
\small
\begin{flalign*}
&\mathbb{P}\Bigg[\frac{1}{N} \sum_{i=1}^N G_1(\bm{y}_i) \bm{1}_{\{G_1(\bm{y}_i) > M_N \}}
> a_4 \frac{\log N}{N} \Bigg]\\
& = \mathbb{P} \Bigg[\frac{1}{N} \sum_{i=1}^N
\Big[c_1+c_2\|\bm{y}_i\|_2+c_3\|\bm{y}_i\|_2^2\Big]\bm{1}_{\{c_1+c_2\|\bm{y}_i\|_2+c_3\|\bm{y}_i\|_2^2>M_N \}} >
a_4\frac{\log N}{N}\Bigg]  , \\
&\textrm{and using Markov's inequality gives}\\
&\le \frac{1}{a_4}\frac{1}{\log N} \Bigg\{c_1 \sum_{i=1}^N \mathbb{P}
\Big[c_1+c_2\|\bm{y}_i\|_2+c_3\|\bm{y}_i\|_2^2>M_N\Big]\\
& \quad + c_2\sum_{i=1}^N
\mathbb{E}\Big[\|\bm{y}_i\|_2\bm{1}_{\{c_1+c_2\|\bm{y}_i\|_2+c_3\|\bm{y}_i\|_2^2>M_N \}} \Big] \\
& \quad + c_3
\sum_{i=1}^N \mathbb{E}\Big[\|\bm{y}_i\|_2^2\bm{1}_{\{c_1+c_2\|\bm{y}_i\|_2+c_3\|\bm{y}_i\|_2^2>M_N
  \}} \Big] \Bigg\}  .\\
& \textrm{For any $0<\varepsilon<1/2$, we employ H\"older's inquality}\\
&\le \frac{1}{a_4}\frac{1}{\log N} \Bigg\{c_1 \sum_{i=1}^N \mathbb{P}
\Big[c_1+c_2\|\bm{y}_i\|_2+c_3\|\bm{y}_i\|_2^2>M_N\Big]\\
& \quad + c_2 \sum_{i=1}^N \mathbb{E}\Big[
(\|\bm{y}_i\|_2)^{\frac{1}{\varepsilon}}\Big] ^{\varepsilon}
\mathbb{P}\Big[c_1+c_2\|\bm{y}_i\|_2+c_3\|\bm{y}_i\|_2^2>M_N\Big]^{1-\varepsilon}\\
& \quad + c_3
\sum_{i=1}^N \mathbb{E}\Big[(\|\bm{y}_i\|_2^2)^{\frac{1}{\varepsilon}}\Big]
^{\varepsilon}
\mathbb{P}\Big[c_1+c_2\|\bm{y}_i\|_2+c_3\|\bm{y}_i\|_2^2>M_N\Big]^{1-\varepsilon}
\Bigg\}  .\\
& \textrm{Since all moments of the non-central $\chi^2$-distribution are
  finite, we get}\\
&\le \frac{1}{a_4}\frac{1}{\log N} \Bigg\{c_1 \sum_{i=1}^N \mathbb{P}
\Big[c_1+c_2\|\bm{y}_i\|_2+c_3\|\bm{y}_i\|_2^2>M_N\Big] \\
& \quad + \tilde{c}_2\sum_{i=1}^N
\mathbb{P}\Big[c_1+c_2\|\bm{y}_i\|_2+c_3\|\bm{y}_i\|_2^2>M_N\Big]^{1-\varepsilon}\\
& \quad + \tilde{c}_3 \sum_{i=1}^N \mathbb{P}\Big[c_1+c_2\|\bm{y}_i\|_2+c_3\|\bm{y}_i\|_2^2>M_N\Big]^{1-\varepsilon}
\Bigg\}  .\\
\end{flalign*}
\begin{flalign*}
&\textrm{With (\ref{041009}) we finally obtain} \\
& \le \frac{2}{a_4}\frac{1}{\log N} \Big\{ c_1 \sum_{i=1}^N \frac{n}{N^2} +
\tilde{c}_2\sum_{i=1}^N \Big(\frac{n}{N^2}\Big)^{1-\varepsilon} +
\tilde{c}_3 \sum_{i=1}^N \Big(\frac{n}{N^2}\Big)^{1-\varepsilon}
\Big\} \\
& \le \frac{\rho}{\log N}\frac{1}{N^{1-2\varepsilon}}  .
\end{flalign*}
\end{proof}
\normalsize
Now, we have shown that the probability (\ref{231209}) in Lemma \ref{Lemma2} is
large. 
Defining the set $\mathcal{J}$ by
\begin{equation}
\mathcal{J} = \Bigg\{ \sup_{\bm{\phi}^T=(\bm\beta^T,\bm{\eta}^T) \in \bm\Phi} \frac{\Big|V_N(\bm{\phi})-V_N(\bm{\phi}_0)
  \Big|}{(\|\bm\beta-\bm\beta_0\|_1 + \|\bm\eta-\bm\eta_0\|_2)\vee \lambda_0} \le T\lambda_0
\Bigg\} \tag{A.6} \label{171209}
\end{equation}
means that $\mathcal{J}$ has large probability. The rest of the
proof of Theorem 1 is as in \cite{Stad09}.

\section*{Appendix B: Proof of Theorem 2}
The proof of the theorem comprises two main parts. First, we have to show
that three conditions presented in \cite{Stad09} are fulfilled. Afterwards we
can present the proof of the theorem.

\subsection*{Appendix B1: Verification of the conditions}
We have to check Conditions $1-3$ in \cite{Stad09}. Subsequently,
each of these is stated as a Lemma and again for simplicity, we drop the
index $i$.\\
Let us introduce a slightly different parametrization, which coincides with
the one in \cite{Stad09} and which simplifies the
proofs below. For $\bm{x}_k \in
\mathbb{R}^p, k=1,\ldots,n$, define
$\bm{X}^T=(\bm{x}_1,\ldots,\bm{x}_n)$. Let
\begin{align*}
 \bm\vartheta^T &=  \bm\vartheta(\bm{X})^T =
 (\bm{x}_1^T\bm\beta,\ldots,\bm{x}_n^T\bm\beta,\bm\theta^T, 2 \log \sigma)
 = ((\bm{X}\bm\beta)^T,\bm\theta^T, \varrho)\\
& = ((\bm{X}\bm\beta)^T,\bm\eta^T )
 =(\bm\xi(\bm{X})^T,\bm\eta^T)
 =(\bm\xi^T,\bm\eta^T) \in \mathbb{R}^{d}
\end{align*}
 be the parameter vector with dimension $d:=n+q^*+1$. By
(12), the parameter space is bounded by the constant $K$: $$\bm\Theta
\subset \{\bm\vartheta \in \mathbb{R}^d: \|\bm\vartheta\|_{\infty}  \le K, \bm\Psi>0\}$$ where
$\|\bm\vartheta\|_{\infty}:=\max_{1 \le j \le d} |\vartheta_j|$.
Let $\{f_{\bm\vartheta}(\bm{y}), \bm\vartheta \in \bm\Theta \}$ be the Gaussian density
of $\bm{y}$ and $\ell_{\bm{\vartheta}}(\bm{y})$ its log-likelihood function. Moreover, let $\bm\vartheta_0$ be the true parameter vector.
\begin{lemma} \label{221109} Under Assumption 1 holds
\begin{equation*}
\sup_{\bm\vartheta \in \bm\Theta} \max_{(j_1,j_2,j_3) \in \{1,\ldots,d \}^3}
\Bigg|\frac{\partial^3}{\partial \vartheta_{j_1} \partial
\vartheta_{j_2} \partial \vartheta_{j_3}} \ell_{\bm{\vartheta}}(\bm{y})  \Big| \le G_2(\bm{y})
 ,
\end{equation*}
where
\begin{equation*}
\sup_{\bm{X} \in \mathcal{X}^n} \int G_2(\bm{y})f_{\bm\vartheta_0}(\bm{y}) d\mu(\bm{y}) \le C_2 <
\infty  .
\end{equation*}
\end{lemma}
\begin{proof} Set $G_2(\bm{y}):=d_1 + d_2\|\bm{y}\|_2 + d_3\|\bm{y}\|_2^2$
for appropriate constants $d_1,d_2,d_3 \in \mathbb{R}_+$. The proof makes use
of the
same techniques as the proof of Lemma \ref{Lemma1} in the  Appendix A. 
\end{proof}

\begin{lemma} \label{231109} Under Assumption 2 (b), the Fisher information matrix $\mathcal{I}(\bm\xi(\bm{X}),\bm\eta)$ is
strictly positive definite.
\end{lemma}
\begin{proof} For $\bm{y} \thicksim \mathcal{N}_n(\bm\xi,\bm{V})$ with
$\bm{V}=\bm{Z}\bm\Psi\bm{Z}^T + e^{\varrho}\bm{I}$, the Fisher information matrix
is given by \citep{McCCS01}
\small
\begin{equation*}
\mathcal{I}(\bm\xi,\bm\eta) =
\begin{pmatrix}
\bm{V}^{-1} & \bm{0} \\
\bm{0} & \frac{1}{2} \{\trace(\bm{V}^{-1}\frac{\partial \bm{V}}{\partial
    \vartheta_r}\bm{V}^{-1}\frac{\partial \bm{V}}{\partial
    \vartheta_s})\}_{r,s=n+1}^{n+q^*+1} \\

\end{pmatrix}
.
\end{equation*}
\normalsize
The upper left part of the matrix is given by $\bm{V}^{-1}$,
which is positive definite. By Assumption 2 (b), the lower right
part is also positive definite, hence we get the claim.
\end{proof}

\begin{lemma} \label{241109} Under Assumption 2 (a), for all $\epsilon > 0$, there exists an $\alpha_{\epsilon}>0$, such that
\begin{equation*}
\inf_{\bm{X} \in \mathcal{X}^n} \inf_{\bm\vartheta \in \bm\Theta, \|\bm\vartheta-\bm\vartheta_0\|_2 >
  \epsilon} \mathcal{E}(\bm\vartheta(\bm{X}) |\bm\vartheta_0(\bm{X})) \ge
\alpha_{\epsilon} .
\end{equation*}
\end{lemma}
\begin{proof} Let $\bm\vartheta^T=(\bm\xi^T,\bm\eta^T)$, 
$\bm\vartheta_0^T=(\bm\xi_0^T,\bm\eta_0^T)$, $\bm{V}$ and $\bm{V}_0$ the
corresponding covariance matrices. Then $\log f_{\bm\vartheta_0}(\bm{y})-\log f_{\bm\vartheta}(\bm{y}) = \frac{1}{2} \log|\bm{V}| -\frac{1}{2}
\log|\bm{V}_0| +\frac{1}{2}(\bm{y}-\bm\xi)^T\bm{V}^{-1}(\bm{y}-\bm\xi)
-\frac{1}{2}(\bm{y}-\bm\xi_0)^T\bm{V}_0^{-1} (\bm{y}-\bm\xi_0)$. Since 
$\mathcal{E}(\bm\vartheta|\bm\vartheta_0):=\mathbb{E}_{\bm\vartheta_0}\Big[\log f_{\bm\vartheta_0}(\bm{y})-\log f_{\bm\vartheta}(\bm{y})\Big]$, it follows
\begin{equation*}
\mathcal{E}(\bm\vartheta|\bm\vartheta_0) = \frac{1}{2}
\Bigg[\log\frac{|\bm{V}|}{|\bm{V}_0|} + \trace(\bm{V}^{-1}\bm{V}_0) +
(\bm\xi_0-\bm\xi)^T\bm{V}^{-1}(\bm\xi_0-\bm\xi)-n \Bigg].
\end{equation*}

By definition of the excess risk $\mathcal{E}(\bm\vartheta | \bm\vartheta_0) \ge
0$. Denote $\bm\eta^T=(\bm\theta^T,\varrho)$ and $\bm\eta_0^{T} =
(\bm\theta_0^{T},\varrho_0)$, then we can detail:
\begin{equation*}
\log\frac{|\bm{V}|}{|\bm{V}_0|} = -\sum_{j=1}^n \log \Bigg(\frac{\omega_j +
  e^{\varrho_0}}{\omega_j + e^{\varrho}}\Bigg) \quad , \quad \trace(\bm{V}^{-1}\bm{V}_0) = \sum_{j=1}^n \frac{\omega_j +
  e^{\varrho_0}}{\omega_j + e^{\varrho}}.
\end{equation*}
Thus, we get
\begin{equation*}
\mathcal{E}(\bm\vartheta | \bm\vartheta_0) = \frac{1}{2} \Big\{(\bm\xi_0-\bm\xi)^T \bm{V}^{-1}(\bm\xi_0-\bm\xi) \Big\} +\frac{1}{2} \sum_{j=1}^n \Bigg\{\frac{\omega_j +
  e^{\varrho_0}}{\omega_j +
  e^{\varrho}} - \log\Bigg(\frac{\omega_j +
  e^{\varrho_0}}{\omega_j +  e^{\varrho}}\Bigg) -1  \Bigg\}.
\end{equation*}
The first term is strictly positive if $\bm\xi_0 \ne \bm\xi$ and zero iff $\bm\xi_0 =
\bm\xi$. The second term is a function of the form $u-\log(u) -1 \ge 0$ for $u\ge 0$. The
second term is only zero if all terms are exactly zero. Due to
Assumption 2 (a), we get the claim.
\end{proof}

\subsection*{Appendix B2: Main proof of Theorem 2}
Let us write $\| \bm{W} \bm\beta \|_1 := \sum_{k=1}^p w_k | \beta_k | $.
Using the definition of $\hat{\bm\phi}$, and on ${\cal J}$,
we have the basic inequality
$$\bar {\cal E} (\hat{\bm\phi} \vert \bm\phi_0) +
\lambda   \| \bm{W} \hat{\bm\beta} \|_1 \le
T \lambda_0 \biggl [  ( \| \hat{\bm\beta} - \bm\beta_0 \|_1 + \|\hat{\bm\eta} - \bm\eta_0 \|_2 
 ) \vee \lambda_0 
\biggr ] + \lambda \| \bm{W} \bm\beta_0 \|_1  .$$

Invoking the triangle inequality $\|\bm{W}\bm\beta_0 \|_1 - 
\| \bm{W}\hat{\bm\beta}_{S_0} \|_1 \le \| \bm{W}( \hat{\bm\beta}_{S_0} - \bm\beta_{0}) \|_1 $,
we obtain
$$  \bar {\cal E} (\hat{\bm\phi} \vert \bm\phi_0 ) 
+ \lambda \| \bm{W} \hat{\bm\beta}_{S_0^c} \|_1  \le 
T \lambda_0 \biggl [  ( \| \hat{\bm\beta} - \bm\beta_0 \|_1 + \|\hat{\bm\eta} - \bm\eta_0 \|_2 
 ) \vee \lambda_0 
\biggr ]+ \lambda \| \bm{W}( \hat{\bm\beta}_{S_0} - \bm\beta_0 ) \|_1 . $$
Since $w_k \ge 1/\delta $ for $k \in S_0^c$ and $w_k \le 1/\delta $ for $k \in S_0$, we arrive at
\begin{equation} \tag{B.1} \label{start1}
  \bar {\cal E} (\hat{\bm\phi} \vert \bm\phi_0 ) 
+ \lambda/\delta  \| \hat{\bm\beta}_{S_0^c} \|_1  \le 
T \lambda_0 \biggl [  ( \| \hat{\bm\beta} - \bm\beta_0 \|_1 + \|\hat{\bm\eta} - \bm\eta_0 \|_2 
 ) \vee \lambda_0 
\biggr ]+ \lambda/\delta  \| \hat{\bm\beta}_{S_0} - \bm\beta_0  \|_1 .
\end{equation} 

By the arguments in \cite{Stad09}, for a constant $c_0$ independent
of $N$, $n$, $p$ and the design,
\begin{equation} \tag{B.2} \label{start2}
\bar {\cal E} (\hat{\bm\phi} \vert \bm\phi_0 ) \ge
( \hat{\bm\beta} - \bm\beta_0 )^T \bm\Sigma_{N,n} ( \hat{\bm\beta} - \bm\beta_0 ) / c_0^2 
+ \| \hat{\bm\eta} - \bm\eta_0\|_2^2 / c_0^2 . 
 \end{equation}
 
 {\bf Case 1} Suppose that
$$\| \hat{\bm\beta} - \bm\beta_0 \|_1 + \| \hat{\bm\eta} - \bm\eta_0 \|_2 \le \lambda_0 . $$
Then we find from (\ref{start1}),
$$ \bar {\cal E} (\hat{\bm\phi} \vert \bm\phi_0 )  \le \bar {\cal E} (\hat{\bm\phi}_{\lambda} \vert \bm\phi_0 ) 
+ \lambda / \delta  \| \hat{\bm\beta}_{S_0^c} \|_1  \le T \lambda_0^2  + \lambda/\delta  \| \hat{\bm\beta}_{S_0} - \bm\beta_0  \|_1 
\le T \lambda_0^2  + \lambda/\delta  \| \hat{\bm\beta} - \bm\beta_0  \|_1$$
and hence
\begin{align*}
\bar {\cal E} (\hat{\bm\phi} \vert \bm\phi_0 ) + 2\lambda/\delta  \|
\hat{\bm\beta} - \bm\beta_0\|_1
&\le
T \lambda_0^2  + 3\lambda /\delta \| \hat{\bm\beta} - \bm\beta_0 \|_1\\
&\le (3\lambda /\delta + T \lambda_0) \lambda_0.
\end{align*}
{\bf Case 2} Suppose that
$$\| \hat{\bm\beta} - \bm\beta_0 \|_1 + \| \hat{\bm\eta} - \bm\eta_0 \|_2 \ge \lambda_0 , $$
and that
$$T \lambda_0 \| \hat{\bm\eta} - \bm\eta_0 \|_2 
\ge ( \lambda /\delta + T \lambda_0 ) \| \hat{\bm\beta}_{S_0} - \bm\beta_{0} \|_1 . $$ 
Then we get, adding $( \lambda/\delta  - \lambda_0 T) \| \hat{\bm\beta}_{S_0} - \bm\beta_0 \|_1$
to left- and right-hand side of (\ref{start1}), 
\begin{eqnarray*}
& &\bar {\cal E} (\hat{\bm\phi}  \vert \bm\phi_0 ) + ( \lambda/\delta - T \lambda_0)
\| \hat{\bm\beta} -\bm\beta_0  \|_1 \le (\lambda/\delta +T\lambda_0)  \| \hat{\bm\eta} - \bm\eta_0 \|_2 \\
&\le& (\lambda/\delta  + T \lambda_0)^2  c_0^2/2+ \| \hat{\bm\eta} - \bm\eta_0 \|_2^2 / (2c_0^2) \\
&\le& (\lambda/\delta  +T\lambda_0)^2 c_0^2 /2+  \bar {\cal E} (\hat{\bm\phi} \vert \bm\phi_0
)/2,
\end{eqnarray*}
where in the last inequality, we applied (\ref{start2}). 
So then
$$ \bar {\cal E} (\hat{\bm\phi}  \vert \bm\phi_0 ) + 2 ( \lambda/\delta  - T \lambda_0)
\| \hat{\bm\beta} - \bm\beta_0 \|_1 \le (\lambda/\delta +T \lambda_0)^2 c_0^2.$$

  {\bf Case 3} Suppose that
$$\| \hat{\bm\beta} - \bm\beta_0 \|_1 + \| \hat{\bm\eta} - \bm\eta_0 \|_2 \ge \lambda_0 , $$
and that
$$T \lambda_0 \| \hat{\bm\eta} - \bm\eta_0 \|_2 
\le ( \lambda/\delta  + T \lambda_0 ) \| \hat{\bm\beta}_{S_0} -\bm\beta_0 \|_1 . $$
Then we have
\begin{equation} \tag{B.3} \label{start3}
\bar {\cal E} (\hat{\bm\phi} \vert \bm\phi_0 ) +  ( \lambda /\delta - T \lambda_0)
\| \hat{\bm\beta}_{S_0^c}  \|_1 \le 2 (\lambda/\delta + T \lambda_0 ) \| \hat{\bm\beta}_{S_0}-
\bm\beta_0\|_1 . 
\end{equation}
Because $\lambda_0 \le \lambda / (2 \delta)$, inequality (\ref{start3}) implies
$$\| \hat{\bm\beta}_{S_0^c} \|_1 \le 6 \| \hat{\bm\beta}_{S_0} - \bm\beta_0 \|_1 . $$
We can therefore apply the restricted
  eigenvalue condition to
$\hat{\bm\beta} - \bm\beta_0 $. 
But first, add $( \lambda/\delta - \lambda_0 T) \| \hat{\bm\beta}_{S_0} - \bm\beta_0 \|_1$
to the left- and right-hand side of (\ref{start3}).
The restricted eigenvalue condition now gives (invoking
$2 ( \lambda/\delta + T \lambda_0) + ( \lambda\delta  - T \lambda_0)
 \le 
3 (\lambda/ \delta + T \lambda_0)$)
\begin{eqnarray*}
& & 
\bar {\cal E} (\hat{\bm\phi} \vert \bm\phi_0 ) +  ( \lambda/\delta  - T \lambda_0)
\| \hat{\bm\beta} - \bm\beta_0 \|_1\\
& \le & 3 (\lambda/\delta + T \lambda_0 )
\sqrt s_0  \| \hat{\bm\beta}_{S_0} - \bm\beta_0 \|_2 \\
&\le& 3 (\lambda/\delta + T \lambda_0 ) \sqrt s_0  \kappa 
\sqrt { (\hat{\bm\beta} - \bm\beta_0 )^T \bm\Sigma_{N,n} (\hat{\bm\beta} - \bm\beta_0) } 
\\
&\le& 9 (\lambda + T \lambda_0 )^2 c_0^2 \kappa^2 s_0 /2 + 
 \bar {\cal E} (\hat{\bm\phi} \vert \bm\phi_0 )/2,
\end{eqnarray*}
applying again (\ref{start2}) in the last step.
So we arrive at 
$$ \bar {\cal E} (\hat{\bm\phi} \vert \bm\phi_0 ) +  2 ( \lambda/\delta - T \lambda_0)
\| \hat{\bm\beta} - \bm\beta_0  \|_1\le 9 (\lambda/\delta + T \lambda_0 )^2 c_0^2 \kappa^2 s_0.$$

\subsection*{Appendix B3: Proof of Corollary 3}
For the estimator in (13) we have:
$$\|\hat{\bm\beta}_{init} - \bm\beta_0 \|_{\infty} \le \|\bm{\bm\beta} -
\bm\beta_0 \|_1 \le \delta_{init}.$$
Consider $k \in S_0$ with $|\beta_{0,k}|>\delta_{init}$. Then it must hold
that $\hat{\beta}_k \ne 0$ (since otherwise, if $\hat{\beta_k}$ were equal
to zero, $\|\hat{\bm\beta} - \bm\beta_0  \|_{\infty} \le |\hat{\beta}_k
  -\beta_{0,k}|=|\beta_{0,k}| > \delta_{init}$ which is a contradiction
  to the $\ell_{\infty}$-estimation error bound). The argument for the
  adaptive $\ell_1$-penalized estimator (14) is analogous.

\section*{Appendix C: Computational details of Algorithm 1}

\indent \textit{(0): Initial value $\bm{\phi}^0$.} As a starting value for $\bm\beta$, we choose an ordinary Lasso solution by cross-validation
ignoring the grouping structure among the observations. By doing so, we
ensure that we are at least as good (with respect to the objective
function) as an ordinary Lasso in a linear model. The calculation of the
starting value for $\bm\theta$ depends on the specific structure of
$\bm\Psi$ and may be performed as in the (Gauss-Seidel) iteration.
The point we would like to make is that those elements that are estimated
as zero in $\bm{\phi}^0$ may escape from zero and non-vanishing elements of
$\bm{\phi}^0$ can be set to zero during Algorithm 1.\\

\textit{(1): Choice of $h^{\ell}$.} For numerical convergence (see Theorem
3), we require that $h^{\ell}$ is positive and bounded. We
use the diagonal elements of the Fisher information $\mathcal{I}(\bm{\phi})$ and, as proposed in \cite{Tseng07}, for constants $c_{min}$ and
$c_{max}$ we set $h^{\ell} = \min(\max(\mathcal{I}(\bm{\phi})_{\mathcal{S}^{\ell}
  \mathcal{S}^{\ell}},c_{min}),c_{max})$ with $c_{min}=10^{-6}$ and
$c_{max}=10^8$ in the \texttt{R} package \texttt{lmmlasso}.\\

\textit{(2): Calculation of $d^{\ell}$.} We have to distinguish whether the index $\mathcal{S}^{\ell}$
appears in $P(\bm\phi)$ or not:
\begin{equation} \tag{C.1}
\footnotesize
d^{\ell} = \begin{cases}
\displaystyle \median
 \Bigg(\frac{\lambda - \frac{\partial}{\partial \phi_{\mathcal{S}^{\ell}}}
   g(\bm{\phi}^{\ell})}{h^{\ell}},-\beta_{\mathcal{S}^{\ell}},\frac{-\lambda-\frac{\partial}{\partial \phi_{\mathcal{S}^{\ell}}}
 g(\bm{\phi}^{\ell})}{h^{\ell}}\Bigg)& \mathcal{S}^{\ell}\in \{1,\ldots,p\},\\
-\frac{\partial}{\partial \phi_{\mathcal{S}^{\ell}}} g(\bm{\phi}^{\ell})/h^{\ell}
&\mbox{else}.
\end{cases}
\end{equation}\\

\textit{(3): Choice of $\alpha^{\ell}$.} The step length $\alpha^{\ell}$ is chosen in such a way that in each step, there
is an improvement in the objective function $Q_{\lambda}(.)$. We use the Armijo
rule which is defined as follows:\\
\textit{Choose $\alpha_{init}^{\ell} >0$ and let $\alpha^{\ell}$ be the largest element of
$\{\alpha_{init}^{\ell}\delta^r\}_{r=0,1,2,..}$ satisfying
\begin{equation*}
Q_{\lambda}(\bm{\phi}^{\ell} +
\alpha^{\ell}d^{\ell}\bm{e}_{\mathcal{S}^{\ell}}) \le
Q_{\lambda}(\bm{\phi}^{\ell}) + \alpha^{\ell} \varrho \triangle^{\ell},
\end{equation*}
 where $\triangle^{\ell} := \partial/ \partial \phi_{\mathcal{S}^{\ell}} g(\bm{\phi}^{\ell})d^{\ell} + \gamma (d^{\ell})^2
h^{\ell} + \lambda P(\bm{\phi}^{\ell}+d^{\ell}\bm{e}_{\mathcal{S}^{\ell}}) - \lambda
P(\bm{\phi}^{\ell})$.}\\

 The choice of the constants comply with the suggestions in
\cite{BerD99} and are $\delta=0.1, \varrho=0.001,\gamma=0$ and
$\alpha_{init}^{\ell}=1$ for all $\ell$.\\

\textit{Simplification of (2) and (3) for the $\bm\beta$-parameter.} If $\mathcal{I}(\bm\phi)_{\mathcal{S}^{\ell}
  \mathcal{S}^{\ell}}$ is not truncated, we take advantage of the fact that
$g(\bm\phi)$ is quadratic with respect to $\bm{\beta}$. Using $\alpha_{init}^{\ell}=1$, the stepsize $\alpha^{\ell}$
chosen by the Armijo rule ($r=0$) leads to the minimum of $g(\bm\phi^{\ell})$ with respect to $\beta_{\mathcal{S}^{\ell}}$. The update
$\hat{\beta}^{\ell+1}_{\mathcal{S}^{\ell}}(\lambda)$ is then given analytically by
\begin{equation} \tag{C.2} \label{111010}
\footnotesize
\hat{\beta}^{\ell+1}_{\mathcal{S}^{\ell}}(\lambda)=\sign\Big((\bm{y}-\tilde{\bm{y}})^T \bm{V}^{-1} \bm{x}_{\mathcal{S}^{\ell}}\Big)
\frac{\Big(|(\bm{y}-\tilde{\bm{y}})^T \bm{V}^{-1}
  \bm{x}_{\mathcal{S}^{\ell}}|-\lambda\Big)_+}{\bm{x}_{\mathcal{S}^{\ell}}^T\bm{V}^{-1} \bm{x}_{\mathcal{S}^{\ell}}} , 
\end{equation}
where $\bm{X}=(\bm{x}_1,\ldots,\bm{x}_p)$, $\tilde{\bm{y}}=\bm{X}^{(-\mathcal{S}^{\ell})}\hat{\bm{\beta}}_{(-\mathcal{S}^{\ell})}^{\ell}$
(leaving out the $\mathcal{S}^{\ell}$th variable), $(.)_+=\max(.,0)$ and
$\sign(.)$ the signum function.\\
Most often, $\mathcal{I}(\bm\phi)_{\mathcal{S}^{\ell} \mathcal{S}^{\ell}}$
is not truncated and hence the analytical formula (\ref{111010}) can be used. This
simplification reduces the computational cost remarkably, especially in
the high-dimensional setup.\\

\textit{Parametrization of $\bm\Psi$.} We parametrize $\bm\Psi$ by a set of unconstrained parameters
$\bm\theta$. A discussion how to parametrize a positive definite
variance-covariance matrix by an unconstrained set of parameters can be found in \cite{PinJB2000} and \cite{Pinh96}. In the current version of the
\texttt{lmmlasso} package we employ the Cholesky decomposition
$\bm\Psi=\bm{LL}^T$ where $\bm\theta$ corresponds to the lower triangular elements of $\bm{L}$.\\

\textit{Choice of the $\lambda$-sequence.} 
We choose a $\lambda_{1}$ sufficiently large such that all penalized
coefficients are zero. We calculate a sequence
$\lambda_{1}>\lambda_{2}>\ldots$ on a log-scale until a model with a certain sparsity level is reached. At 
latest, we stop if the number of selected fixed-effects variables is larger than the
total number of observations. The optimal $\lambda$ is then chosen by
\begin{equation*}
\lambda_{opt}= \argmin_{k \ge 1} BIC_{\lambda_{k}}.
\end{equation*}

\textit{Active-Set Algorithm.}
Assuming that the solution is sparse, we can reduce the computing time by
using an active-set algorithm, which is used in \cite{Meie08} and
\cite{Fried08}. More specifically, we do not cycle through all coordinates,
but we restrict ourselves to the current active set $S(\hat{\bm{\beta}})$ and update all
coordinates of $\hat{\bm{\beta}}$ only every $D$th iteration. This reduces
the computational time considerably.\\

\textit{Proof of Theorem 3.} For the precise definition of cluster and
stationary point we refer to \cite{Tseng07}. It remains to check that the
assumptions in \cite{Tseng07} are fulfilled. More precisely: $\lambda>0$,
$P(.)=|.|_1$ is a proper, convex, continuous function and block-separable
with respect to $\mathcal{S}^{\ell}$, $g(.)$ is continuously differentiable on $dom(P)=\{\bm{\phi}
  |P(\bm{\phi}) < \infty \}$, $c_{min}  \le h^{\ell} \le c_{max}$ for $\ell \ge 0$
  and $0 < c_{min} \le c_{max}$. Moreover, $\sup_{\ell} \alpha^{\ell}>0$ and
  $\inf_{\ell} \alpha_{init}^{\ell} > 0$.

\section*{Appendix D: Simulation study for the low-dimensional setting}
In this
setting, we will compare \textit{lmmLasso} and \textit{lmmadLasso} with the
classical linear mixed-effects framework (\textit{lme}) from \cite{PinJB2000} and both the Lasso and
the adaptive Lasso. The optimal model for the \textit{lme} procedure is
determined by backward elimination.

The two examples are chosen in the following way ($\beta_{0,1}=1$ is
the unpenalized intercept):
\begin{itemize}
\item [$L_1$:]  $N=25$,  $n=6$, $N_T=150$, $p=10$, $q=3$, $\sigma^2=0.25$, $\theta^2=0.56$ and
  $s_0=5$ with $\bm\beta_0=(1,2,4,3,3,0,\ldots,0)^T$. 
\item [$L_2$:]  $N=30$, $n=6$, $N_T=180$, $p=15$, $q=3$,
  $\sigma^2=0.25$,
\begin{equation*}
\bm\Psi=\begin{pmatrix} 5 &2 & 0.5 \\ 2 & 2& 1 \\ 0.5 & 1 &  1 \end{pmatrix}
\end{equation*}
and $s_0=5$ with $\bm\beta_0=(1,2,4,3,3,0,\ldots,0)^T$.
\end{itemize}

The results in the form of means and standard deviations (in parentheses) over 100 simulation
runs are reported in Table \ref{table1}, \ref{table2} and \ref{table3}. Therein, $|S(\hat{\bm{\beta}})|$
denotes the cardinality of the estimated active set and TP is the number of true
positives. 
We would like to emphasize that we do not penalize any covariate having a
random-effects coefficient (indicated by an asterisk $^*$).
\begin{table}[!h]
\footnotesize
\begin{center}
\caption{\textit{Comparison of lmmLasso, lmmadLasso, lme, Lasso and adLasso
    for model $L_1$}} \label{table1}
\vspace{0.2cm}
\begin{tabular}{c||cc|cc|ccccc}
\hline \hline
 Method & $|S(\hat{\bm\beta})|$ & TP & $\hat{\sigma}^2$ & $\hat{\theta}^2$
& $\hat{\beta}_1$ & $\hat{\beta}_{2}$ &  $\hat{\beta}_{3}$ & $\hat{\beta}_{4}$ & $\hat{\beta}_{5}$ \\ 
\hline
     true& 5 & 5  & 0.25 &  0.56 & 1 & 2 & 4 & 3 & 3 \\
\hline
lmmLasso &5.94  & 5 & 0.24 & 0.55&$ 0.99^*$ &$2.01^*$ &$4.03^*$& 2.94 & 2.95 \\
&(1.04)&(0)&(0.04)&(0.11)&(0.14)&(0.15)& (0.15) & (0.06)& (0.06)\\
lmmadLasso &5.11  & 5 & 0.24 & 0.55&$ 0.99^*$ &$2.01^*$&$4.02^*$& 2.99 & 3 \\
&(0.31)&(0)&(0.04)&(0.11)&(0.14)&(0.15)& (0.15) & (0.05)& (0.06)\\
lme& 5.14 & 5 & 0.24& 0.55 &$0.99^*$ &$2.01^*$&$4.02^*$&$2.99^*$&$3^*$\\
   &(0.35)&(0)&(0.04)&(0.11)&(0.14)&(0.15)& (0.15) & (0.05)& (0.06)\\
Lasso& 5.54&5&1.85& - &$1.00^*$ &$1.99^*$&$4.04^*$&2.88 & 2.89\\
   &(0.69)&(0)&(0.38)& -&(0.16)&(0.20)& (0.18) & (0.12)& (0.12)\\
adLasso& 5.54 & 5 &1.81& - &$1.00^*$ &$1.99^*$&$4.01^*$&2.99 & 3.00\\
   &(0.69)&(0)&(0.37)&-&(0.16)&(0.20)& (0.18) & (0.11)& (0.11)\\
\hline
\end{tabular}
\vspace{0.4cm}

\scriptsize
* indicates that the corresponding fixed-effects coefficient is not subject to penalization
\end{center}
\end{table}

\begin{table}[!h]
\footnotesize
\begin{center}
\caption{\textit{Comparison of lmmLasso, lmmadLasso, lme, Lasso and adLasso for model $L_2$}} \label{table2}
\vspace{0.2cm}
\begin{tabular}{c||cc|c|ccccc}
\hline \hline
Method & $|S(\hat{\bm\beta})|$ & TP & $\hat{\sigma}^2$ 
& $\hat{\beta}_1$ & $\hat{\beta}_{2}$ &  $\hat{\beta}_{3}$ & $\hat{\beta}_{4}$ & $\hat{\beta}_{5}$\\ 
\hline
     true& 5 & 5  & 0.25 & 1 & 2 & 4 & 3 & 3\\
\hline
lmmLasso &7.33  & 5 & 0.24 &$1.00^*$&$1.96^*$&$3.99^*$& 2.95 & 2.94\\
&(1.54)&(0)&(0.04)&(0.42)&(0.24)&(0.18)&(0.05) & (0.06)\\
lmmadLasso &5.31  & 5 & 0.24 &$1.00^*$&$1.96^*$&$3.98^*$& 3 & 2.99\\
&(0.72)&(0)&(0.04)&(0.42)&(0.24)&(0.18)& (0.05) & (0.06)\\
lme&4.85  &4.75  &0.24 &$0.73^*$&$1.86^*$&$3.95^*$&$3^*$&$2.99^*$\\
&(0.75) &(0.5)&(0.04)&(0.67)&(0.32)&(0.19)&(0.05)&
(0.06)\\
Lasso&5.59 &5  &8.43 &$1.00^*$&$1.92^*$&$4.05^*$&2.72 & 2.68\\
&(1.02) &(0)&(2.27)&(0.44)&(0.38)&(0.29)&(0.27)&(0.24)\\
adLasso&5.59&5&8.23&$1.00^*$&$1.92^*$&$3.98^*$& 2.99 & 2.94\\
&(1.02)&(0)&(2.21)&(0.44)&(0.37)&(0.30)&(0.26)&(0.23)\\
\hline
\end{tabular}
\vspace{0.1cm}

\scriptsize
* indicates that the corresponding fixed-effects coefficient is not subject to penalization
\end{center}
\end{table}

\begin{table}[!h]
\footnotesize
\begin{center}
\caption{\textit{Covariance estimates of lmmLasso and lme for $L_2$}} \label{table3}
\vspace{0.2cm}
\begin{tabular}{c|cccccc}
\hline
Method & $\Psi_{11}$&$\Psi_{12}$&$\Psi_{13}$&$\Psi_{22}$&$\Psi_{23}$&$\Psi_{33}$ \\ 
\hline
true & 5 & 2 & 0.5 & 2 & 1 & 1  \\
\hline
lmmLasso & 4.83 & 1.95 & 0.58 & 1.91 & 1.03 & 1.04\\
         &(1.26)&(0.76)&(0.51)&(0.58)& (0.38) & (0.32)\\
lmmadLasso & 4.84 & 1.95 & 0.58 & 1.92 & 1.04 & 1.04\\
           &(1.26)&(0.76)&(0.51)&(0.58)&(0.39)&(0.32)\\
lme&5.03 &2.01&0.6&1.94&1.04&1.04\\
   &(1.43)&(0.87)&(0.53)&(0.60)&(0.39)&(0.33)\\
\hline
\end{tabular}
\vspace{0.3cm}
\end{center}
\end{table}

We see from the tables that the estimated average active set is sparse and
only slightly larger than the cardinality of the true active set
$S_0=S(\bm\beta_0)$. This property might be expected because it is known
from linear regression that the BIC selects a sparse model. All methods
except \textit{lme} in model $L_2$ are including the true
non-zero coefficients in the active set. Concerning variance components, we
clearly see that the estimated error variance of \textit{Lasso} and
\textit{adLasso} can be reduced and split into the within and
between-subject variability by \textit{lmmLasso} and \textit{lmmadLasso},
respectively. The tables show that the penalized fixed-effects coefficients from
\textit{lmmLasso} have a bias. However, it is smaller than that of the
corresponding coefficients from the \textit{Lasso}. By using
\textit{lmmadLasso}, we can attenuate the bias problem. From Table
\ref{table3} we note that the variance component estimates of
$\bm\Psi$ are
underestimated compared to the results from \textit{lme}. However, by a
closer look, the
fixed effects of \textit{lme} have a larger bias in \textit{lme} than in
\textit{lmmLasso} and \textit{lmmadLasso}. It seems that \textit{lme}
estimates the variance components more precisely while underestimating the
corresponding fixed effects.
Finally, it must be recognized
that the backward selection used for lme regularly breaks down due to convergence problems within the R-function \texttt{lme}.  

\end{document}